\documentclass[sigplan,10pt,table,xcdraw,screen]{acmart}
\renewcommand\footnotetextcopyrightpermission[1]{} 
\usepackage[utf8]{inputenc}
\usepackage{amsmath,amsfonts,amsthm}
\usepackage{graphicx,float,wrapfig}
\usepackage[export]{adjustbox} 
\usepackage{braket}
\usepackage{fancyvrb,xparse,listings}
\usepackage{algorithmic}
\usepackage[ruled,vlined]{algorithm2e}
\usepackage{xcolor}
\usepackage{multirow}
\usepackage{tabularx}
\usepackage{subcaption}
\usepackage{mathtools}

\DeclarePairedDelimiter\floor{\lfloor}{\rfloor}

\graphicspath{{figures/}}

\AtBeginDocument{%
  \providecommand\BibTeX{{%
    \normalfont B\kern-0.5em{\scshape i\kern-0.25em b}\kern-0.8em\TeX}}}

\setcopyright{none}

\acmConference[A Preprint]{}{December 23}{2020}
\acmYear{2020}

\newcommand{\mattwo}[4]{\left[ \begin{array}{cc}#1 & #2 \\ #3 & #4\end{array}\right]}

\newcounter{qasmindex}

\NewDocumentEnvironment{qprognamed}{ m }
    {\VerbatimOut{qasm/qasm/#1.qasm}}
    {\endVerbatimOut
    \immediate\write18{cd qasm && bash build_qasm.sh #1}
    \includegraphics[valign=c, max width=\textwidth, max height=0.2\textheight]{qasm/out/#1.pdf}
    }

\NewDocumentEnvironment{qprog}{} 
{\VerbatimOut{qasm/qasm/\theqasmindex-inline.qasm}}
{\endVerbatimOut
\immediate\write18{cd qasm && bash build_qasm.sh \theqasmindex-inline}
\includegraphics[valign=c, max width=\textwidth, max height=0.2\textheight]{qasm/out/\theqasmindex-inline.pdf}
\stepcounter{qasmindex}
}

\newtheorem{theorem}{Theorem}
\newtheorem{lemma}{Lemma}
\newtheorem{definition}{Definition}
\newtheorem{example}{Example}
\newtheorem{note}{Note}
\newcommand{\mycomment}[1]{}

\begin{document}
\immediate\write18{mkdir -p qasm/qasm}
\immediate\write18{mkdir -p qasm/out}
\immediate\write18{mkdir -p qasm/build}
\title{Software Pipelining for Quantum Loop Programs}


\author{Guo Jingzhe}
\orcid{0000-0002-7921-9771}
\affiliation{%
  \department{Department of Computer Science}
  \department{and Technology}             
  \institution{Tsinghua University}
  \country{China}                   
}
\email{gjz20@mails.tsinghua.edu.cn}

\author{Mingsheng Ying}
\affiliation{
  \institution{CQSI, FEIT, University of Technology Sydney}
  \country{Australia}                   
}
\affiliation{
  \department{Institute of Software, CAS, China}             
}
\affiliation{
  \institution{Tsinghua University, China}           
}
\email{mingshengying@gmail.com}         

\renewcommand{\shortauthors}{Guo, et al.}

\begin{abstract}
We propose a method for performing software pipelining on quantum for-loop programs, exploiting parallelism in and across iterations.
We redefine concepts that are useful in program optimization, including array aliasing, instruction dependency and resource conflict, this time in optimization of quantum programs.
Using the redefined concepts, we present a software pipelining algorithm exploiting instruction-level parallelism in quantum loop programs.
The optimization method is then evaluated on some test cases, including popular applications like QAOA, and compared with several baseline results. The evaluation results show that our approach outperforms loop optimizers exploiting only in-loop optimization chances by reducing total depth of the loop program to close to the optimal program depth obtained by full loop unrolling, while generating much smaller code in size. This is the first step towards optimization of a quantum program with such loop control flow as far as we know.

\end{abstract}




\maketitle

\section{Introduction}

Quantum computer hardware has reached the so-called quantum supremacy showing that quantum computation can actually outperform classical computation for certain tasks, but it is still in the NISQ (Noisy-Intermediate-Scale-Quantum) era where there are no sufficient quantum bits (qubits, for short) for quantum error correction.

\textbf{Program optimization} is particularly important for executing a quantum program on NISQ hardware in order to reduce the number of required qubits, the length of gate pipeline, and to mitigate quantum noise. Indeed, there has already been plenty of work on optimization and parallelization of quantum programs.
Theoretically, it was proved in \cite{botea2018complexity} that compilation of quantum circuits with discretized time and parallel execution can be NP complete. Practically, 
quantum hardware architectures, especially those based on superconducting qubits, provide instruction level support for exploiting parallelism in quantum programs; for example, Rigetti's Quil \cite{smith2016practical} allows programmers to explicitly specify multiple instructions that do not involve same qubits to be executed together, while in Qiskit, ASAP or ALAP scheduling is performed implicitly \cite{qiskit_basic_scheduler}. Furthermore, several compilers have been implemented that can optimize quantum circuits by exploiting instruction level parallelism; for example, ScaffCC \cite{javadiabhari2015scaffcc} introduces critical path analysis to find the ``depth'' of a quantum program efficiently, revealing how much parallelism there is in a quantum circuit; commutativity-aware logic scheduling is proposed in \cite{shi2019optimized} to adopt a more relaxing quantum dependency graph than ``qubit dependency'' by taking in mind commutativity between the $R_Z$ gates and $\mathit{CNOT}$ gates as well as high-level commutative blocks while scheduling circuits. There are also some more sophisticated optimization strategies reported in in previous works \cite{DBLP:conf/asplos/MuraliMMJ20, quantum_ai_team_and_collaborators_2020_4062499, Sivarajah_2020, Guerreschi_2018} .

Quantum hardware will soon be able to execute quantum programs with more complex program constructs, e.g. for-loops. However, most of the optimization techniques in previous work only deal with sequential quantum circuits. Some methods allow loop programs as their input, but those loops will be unrolled immediately and optimization will be performed on the unrolled code. Loop unrolling is the technique that allows optimization across all iterations of a loop, but comes at a price of long compilation time, redundant final code and run-time compulsory cache misses. As quantum hardware in the near future may allow up to hundreds of qubits, it will often be helpful to preserve loop structure during optimization since the growth in number of qubits will also lead to increment in total gate count, as well as increment in difficulty unrolling the entire program.

\textbf{Software pipelining} \cite{lam1988software} is a common technique in optimizing classical loop prosgrams. Inspired by the execution of an unrolled loop on an out-of-order machine, software pipelining reorganizes the loop by a software compiler instead of by hardware. There are two major approaches for software pipelining: 
\begin{itemize}\item Unrolling-based software pipelining 
usually unrolls loop for several iterations and finds repeating pattern in the unrolled part; see for example \cite{aiken1988optimal}. 
\item Modulo scheduling
guesses an initiation interval first and try to schedule instructions one by one under dependency constraints and resource constraints; see for example \cite{lam1988software}.
\end{itemize}

\textbf{Our Contributions}: We hereby presents a software pipelining algorithm for parallelizing a certain kind of quantum loop programs. Our parallelization technique is based on a novel and more relaxed set of dependency rules on a CZ-architecture (Theorems \ref{rule-1} and \ref{th:generalized_conjugation}). The algorithm is essentially a combination of unrolling-based software pipelining and modulo scheduling \cite{lam1988software}, with several modifications to make it work on quantum loop programs. 

We 
carried out experiments on several 
examples and compared the results with the baseline result obtained by loop unrolling. Our approach proves to be a steady step toward bridging the gap between optimization results without considering across-loop optimization and fully unrolling results while restraining the increase in code size.

\textbf{Organization of the Paper}: In Section \ref{sec:defs_and_notations}, we review some basic definitions used in this paper. The theoretical tools for defining and exploiting parallelism in quantum loop program are developed in Section \ref{sec:approach_optimal}. In Section \ref{sec:resched_loop_body}, we present our  approach of rescheduling instructions across loops, extracting prologue and epilogue so that depth of the loop kernel can be reduced. 
The evaluation results of our experiments are given in Section \ref{sec:evaluation}. The conclusion is drawn in the Section \ref{sec:conclusion}. [\textit{For conciseness, all proofs are given in the Appendices}.]

\section{Preliminaries and Examples}
\label{sec:defs_and_notations}
This section provides some backgrounds \cite{10.5555/1972505, YING201611} on quantum computing and quantum programming.

\subsection{Basics of quantum computing}

The quantum counterparts of bits are qubits. Mathematically, a state of a single qubit is represented by a $2$-dimensional complex column vector $(\alpha, \beta)^T$, where $T$ stands for transpose. It is often written in the Dirac's notation as $\ket{\psi}=\alpha \ket{0}+\beta \ket{1}$ with  $\ket{0}=(1,0)^T$, $\ket{1}=(0,1)^T$ corresponding to classical bits $0$ and $1$, respectively. It is required that $\ket{\psi}$ be unit: $\|\alpha\|^2+\|\beta\|^2=1$. Intuitively, the qubit is in a superposition of $0$ and $1$, and when measuring it, we will get $0$ with probability $\|\alpha\|^2$ and $1$ with probability $\|\beta\|^2$.  
    A gate on the qubit is then modelled by a $2\times 2$ complex matrix $U$. The output of $U$ on an input $\ket{\psi}$ is quantum state $\ket{\psi^\prime}$. Its mathematical representation as a vector is obtained by the ordinary matrix multiplication $U\ket{\psi}$. To guarantee that $\ket{\psi^\prime}$ is always unit, $U$ must be unitary in the sense that $U^\dag U=I$ where $U^\dag$ is the adjoint of $U$ obtained by transposing and then complex conjugating $U$.
    In general, a state of $n$ qubits is represented by a $2^n$-dimensional unit vector, and a gate on $n$ qubits is described by a $2^n\times 2^n$ unitary matrix.
[\textit{For convenience of the readers, 
we present the basic gates used in this paper in Appendix \ref{appendix:basic_quantum_gates}}.]


\subsection{Quantum execution environment}
\label{subsec:arch}
Software pipelining is a highly machine-dependent approach of optimization. So we must give out some basic assumptions about the underlying machine that our algorithm requires. State-of-the-art universal quantum computers differ in many ways:
\begin{itemize}
    \item \textbf{Instruction set}: A quantum computer chooses a universal set of quantum gates as its low-level instructions. For example, IBM Q\cite{Qiskit-Textbook} uses controlled-NOT $\mathit{CNOT}$ and three one-qubit gates $U_1, U_2, U_3$, but Rigetti Quil\cite{smith2016practical} uses controlled-Z $\mathit{CZ}$ and one-qubit rotations $R_X, R_Z$. We use the universal gate set $\left\{U_3, CZ \right\}$ for the reason that $U_3$ itself is universal for single qubit gates, which allows us to merge single qubit gates at compile time. [\textit{see Appendix \ref{appendix:basic_quantum_gates} for the definition of these gates}.]
    \item \textbf{Instruction parallelism}: Different quantum computers are implemented on different technologies, constraining their power to execute multiple instructions simultaneously. Usually superconductive quantum computers support parallelism while ion-trap ones do not. We assume qubit-level parallelism: instructions on different qubits can always be executed simultaneously. 
    \item \textbf{Timing}: Different quantum computers may use different timing strategies, using continuous time or discrete time. Also execution time of different instructions may differ and is highly machine-dependent. Usually a two-qubit gate (e.g. $\mathit{CZ}$ and $\mathit{CNOT}$) costs much longer time than single qubit gates. We use a discrete time model with every gate requiring $1$ tick equally.
    \item \textbf{Qubit connectivity}: Different machines may have different qubit topologies. However, we assume that all gates in the input are directly executable, which may require a layout synthesis step before our optimization.
    \item \textbf{Classical control}. The support for classical control flow varies among different quantum computers; for example, IBM Q does not support any complex control flow, while Rigetti Quil supports branch statements. We assume such classical controls [\textit{see Appendix \ref{appendix:output_language}}].
\end{itemize}



The above assumptions do not fit into the existing quantum hardware architecture perfectly (for instance, IBM Q requires $\mathit{CNOT}$ and Quil disallows $U_3$), while the architecture of Google's devices\cite{quantum_ai_team_and_collaborators_2020_4062499} fits these requirements most. With some slight modifications, however, our method can be easily adapted to unsupported architectures [\textit{see Appendix \ref{sec:adapt_to_arch}}]. 

\subsection{Quantum loop programs}

We focus on a special family of quantum loop programs, called \textit{one-dimensional  for-loop programs}, defined as below:
\begingroup
\allowdisplaybreaks
\begin{align*}
    \textbf{program} :=& \textbf{header}\ \textbf{statement}*\\
    \textbf{header} :=& [(\textbf{qdef}\ |\ \textbf{udef})^*]\\
    \textbf{qdef} :=& qubit\ \textbf{ident}[\mathbb{N}];\\
    \textbf{udef} :=& defgate\ \textbf{ident}[\mathbb{N}] = \textbf{gate};\\
    \textbf{gate} :=& [(\mathbb{C}^{2\times 2})^*]\ |\ R_Z\ |\ R_Z^+\ |\ Unknown\\
    \textbf{gateref} :=& \textbf{ident}[\textbf{expr}]\\
    \textbf{qubit} :=& \textbf{ident}[\textbf{expr}]\\
    \textbf{op} :=& SQ(\textbf{gateref})\ \textbf{qubit}\ 
                       |\ CZ\ \textbf{qubit},\textbf{qubit};\\
    \textbf{statement} :=& \textbf{op}\ 
                        |\ for\ \textbf{ident}\ in\ \mathbb{Z}\ to\ \mathbb{Z} \{\textbf{op}^*\}\\
                       &\quad\ |\  for\ \textbf{ident}\ in\ \textbf{ident}\ to\ \textbf{ident} \{\textbf{op}^*\}\\
    \textbf{expr} :=& \mathbb{Z}  * \textbf{ident} + \mathbb{Z}
\end{align*}
\endgroup
where: 
\begin{itemize}
    \item The loop involves a group of one-dimensional qubit array variables defined by \textbf{qdef}.
    \item The loop has only one iteration variable $i$ starting from $a$ to $b$ with stride $1$. The range $[a, b]$ is completely known at compile time, or completely unknown until execution. This allows our algorithm to be performed on a program with parametric loop range.
    \item All array index expressions are in the form $(ki+b)$, where $i$ is the iteration variable, and $k,b\in \mathbb{Z}$ are known constants.
    \item All operations in the loop body are either an one-qubit gate, or a $CZ$ gate on two qubits. We don't consider measurement operations.
    \item One-qubit gates are defined by \textbf{udef}. They are given as known matrices, or ``an element in an array of unknown matrices'' when a hint on whether the matrix array is diagonal or antidiagonal can be given. This allows our algorithm to be performed on a program with parametric gates or performing different gates on different iterations.
\end{itemize}
At the very start of the entire program, all qubit arrays are initialized as $\ket{0}$. 
Our optimization may introduce some branch statements if the endpoints $a$ and $b$ are unknown before code execution. As a result, the output language of the compiler is a superset of the input language above, with support for branch statements [\textit{see Appendix \ref{appendix:output_language} for one possible definition of output language}].
To show versatility of the above loop, let us consider several popular quantum algorithms.   

\begin{example} 
Grover algorithm \cite{10.1145/237814.237866} is designed for the black-box searching problem: given a function $f: \left\{0,1\right\}^n\rightarrow\left\{0,1\right\}$, find a bitstring $x: \left\{0,1\right\}^n$ such that $f(x)=1$. While a classical algorithm requires $\Omega(n)$ calls to the oracle, Grover search can find a solution in $O(\sqrt{n})$ calls of quantum oracle $U_f(\ket{x}\otimes\ket{b})=\ket{x}\otimes\ket{b\oplus f(x)}$.
This is done by repeating a series of quantum gates, called Grover iteration. Grover search can be written as the loop program: 
\begin{algorithmic}
\FOR{i in  0 \TO N-1 } {
    \STATE {$H[q[i]]$}
} \ENDFOR
\STATE {$H[q_{work}]$}
\FOR{i in  1 \TO $O(\sqrt{N})$ } {
    \STATE {$U_f[q,q_{work}]$; $(2\ket{\psi}\bra{\psi}-I)[q]$}
} \ENDFOR
\end{algorithmic}
\end{example}

{
\begin{example}

A Quantum Approximate Optimization Algorithm (QAOA for short) is designed in  \cite{farhi2014quantum} to solve the MaxCut problem on a given graph $G=\langle V, E\rangle$.
It can be written as a parametric quantum loop program:
\begin{algorithmic}
\FOR{i=0 \TO (N-1) }  {
    \STATE {$H[q[i]]$}
} \ENDFOR
\FOR{i=1 \TO p }  {
    \FOR{$(a, b) \in E$}  {
        \STATE {$CNOT[q[a], q[b]]$; $U_B[i][q[b]]$; $CNOT[q[a], q[b]]$}
    }\ENDFOR
    \FOR{j=0 \TO (N-1) }  {
        \STATE {$U_C[j][q[j]]$}
    } \ENDFOR
} \ENDFOR
\end{algorithmic}
Here, we use parametric gate arrays $U_C[i]=R_X(\beta_i,j)$ and $U_B[i]=R_Z(-\omega_{ab}\gamma_i)$ of rotations. The two innermost loops can be unrolled to satisfy our input language requirements.
Since QAOA repeatedly executes the circuit but each time with different sets of angles $\left\{\beta_i\right\}$ and $\left\{\gamma_i\right\}$, an optimizer has to support compilation of the circuit above without knowing all parameters in advance.
Note that the compiler can know in advance that $U_B[i]$ are diagonal matrices, and this hint might be used during optimization. [for a further explanation of QAOA see \textit{see Appendix \ref{more-examples}}]
\end{example}
}

\section{Theoretical tools}\label{sec:approach_optimal}
In this section, we develop a handful of theoretical techniques required in our optimization. To start, let us identify some of the most critical challenges in optimizing quantum loop programs:
\begin{itemize}
    \item Instructions may be \textbf{merged} together at compile time, potentially reducing the total depth. However, merging instructions needs to know which instructions may be adjacent in the unrolled pattern, thus requiring us to resolve all possible \textbf{qubit aliasings}.
    \item \textbf{Data dependency} graph in a quantum program is usually much denser than that in classical program, since generally two matrices are not commutable, that is, $AB\neq BA$.
    \item \textbf{Resource constraint}, which prevents instructions that do not have dependency from executing together, is quite different in quantum case from classical case.
\end{itemize}
We will show how much optimization can be done by mitigating these challenges in loop reordering.
\subsection{Gate merging}

Our assumptions allow several instructions to be merged into a single instruction with the same effect:
\begin{itemize}
    \item Two adjacent one-qubit gates on the same qubit can be merged, since we are using $U_3$.
    \item Two adjacent $CZ$ gates on the same qubits can cancel each other.
\end{itemize}

\begin{example}
Figure \ref{fig:sq_merge} is a simple case for periodical gate merging pattern. The two one-qubit gates in different iterations may merge with each other, thus simplifying the dependency graph and introducing new opportunities for optimization.
    \begin{figure}
        \begin{subfigure}{0.4\textwidth}
            \begin{lstlisting}[frame=single]
for i=0 to 3 do
    U q[i]; V q[i+1]; W q[i+2];
end for
            \end{lstlisting}
            \caption{Loop program.}
        \end{subfigure}
        \begin{subfigure}{0.2\textwidth}
            \begin{qprognamed}{merge-1}
                qubit q0
                qubit q1
                qubit q2
                qubit q3
                qubit q4
                qubit q5
                def gU,0,'U'
                def gV,0,'V'
                def gW,0,'W'
                gU q0
                gV q1
                gW q2
                nop q3
                nop q4
                nop q5
                gU q1
                gV q2
                gW q3
                nop q4
                nop q5
                gU q2
                gV q3
                gW q4
                nop q5
                gU q3
                gV q4
                gW q5
            \end{qprognamed}
            \caption{Unrolled circuit.}
        \end{subfigure}
        \hspace{1cm}
        \begin{subfigure}{0.11\textwidth}
            \begin{qprognamed}{merge-2}
                qubit q0
                qubit q1
                qubit q2
                qubit q3
                qubit q4
                qubit q5
                def gU,0,'U'
                def gV,0,'V'
                def gW,0,'W'
                def gUV,0,'UV'
                def gUVW,0,'UVW'
                def gVW,0,'VW'
                gU q0
                gUV q1
                gUVW q2
                gUVW q3
                gVW q4
                gW q5
            \end{qprognamed}
            \caption{Merged.}
        \end{subfigure}
        \caption{Single qubit gates can be merged periodically.}
        \label{fig:sq_merge}
    \end{figure}
\end{example}

Gate merging allows us to decrease count of gates, and thus reduce total execution time. However, the existence of potential aliasing adds to the difficulty of finding ``adjacent'' pairs of gates. Figuring out pairs of gates that can be safely merged is one of the critical problems when scheduling the program.

\begin{example}
    Even for a simple program, it can be hard to decide whether two adjacent instructions on a qubit can be merged. Consider the simple program:
    \begin{algorithmic}
    \FOR{i=a \TO b }  {
        \STATE{$H[q[0]]$; $CZ[q[i],q[i+1]]$; $H[q[0]]$;}
    } \ENDFOR
    \end{algorithmic}
    We can merge the Hadamard gates if and only if \ 
       $\forall i, i\neq 0 \wedge (i+1)\neq 0.$
    \begin{figure}
    
        \begin{subfigure}{0.2\textwidth}
    \centering
    \begin{qprognamed}{alias-c1}
            qubit i
            qubit o
            qubit j
            H o
            nop i
            ZZ i,j
            nop o
            H o
    \end{qprognamed}
    \caption{$i\neq0 \wedge i\neq -1$}
    \end{subfigure}
    \begin{subfigure}{0.2\textwidth}
    \centering
    \begin{qprognamed}{alias-c2}
        qubit o
        qubit j
        H o
        ZZ o,j
        H o
    \end{qprognamed}
    \caption{$i=0$}
    \end{subfigure}
    \begin{subfigure}{0.2\textwidth}
    \centering
    \begin{qprognamed}{alias-c3}
        qubit i
        qubit o
        H o
        ZZ i,o
        H o
    \end{qprognamed}
    \caption{$i=-1$}
    \end{subfigure}
            \caption{The $CZ$ gate prevents the two Hadamard gates from merging, due to potential qubit aliasing.}
            \label{fig:qubit_aliasing_problem}
        \end{figure}
Three possible cases of $i$ lead to three different results, as Figure \ref{fig:qubit_aliasing_problem} shows.
    \end{example}

The above example reveals that resolving qubit aliasings is crucial in gate merging.

\subsection{Qubit aliasing resolution}
\label{subsec:qubit_alias}
Allowing arbitrary linear expressions being used to index qubit arrays introduces the problem of qubit aliasing both in a single iteration and across iterations. Potential aliasing in quantum programs leads two kinds of problems: \textit{lack of periodic features in unrolled schedule}, and \textit{extra complexity in detecting aliasings}.

The first problem is that non-periodic features cannot be captured using software-pipelining (or other loop scheduling methods). For example, in Figure \ref{fig:alias_cause_lack_period}, the situation where $CZ$ blocks two Hadamards from merging only occurs in one or two iterations of the loop program, but it prevents the merging in \textbf{all} iterations, since software pipelining can only generate a periodic pattern and has to generate conservative code. The only kind of aliasing (two different qubit expressions refering to the same qubit) that software pipelining can capture is those expressions on the same qubit array and with the same slope, as shown in Figure \ref{fig:alias_captured}.
\begin{figure}
    \begin{subfigure}{0.3\textwidth}
        \begin{lstlisting}[frame=single]
for i=0 to 3 do
    H q[1];
    CZ q[i], q[i+1];
    H q[1];
end for
        \end{lstlisting}
        \caption{Loop program.}
    \end{subfigure}
    \begin{subfigure}{0.3\textwidth}
        \begin{qprognamed}{alias-unroll}
            qubit q0
            qubit q1
            qubit q2
            qubit q3
            qubit q4
            H q1
            ZZ q0,q1
            H q1
            H q1
            ZZ q1,q2
            H q1
            H q1
            ZZ q2,q3
            H q1
            H q1
            ZZ q3,q4
            H q1
        \end{qprognamed}
        \caption{Unrolled circuit.}
    \end{subfigure}
    \caption{Unrolled loop does not reveal periodic feature due to qubis aliasing.}
    \label{fig:alias_cause_lack_period}
\end{figure}
\begin{figure}
    \begin{subfigure}{0.3\textwidth}
        \begin{lstlisting}[frame=single]
for i=0 to 3 do
    H q[i];
    CZ q[i], q[i+1];
    H q[i+1];
end for
        \end{lstlisting}
        \caption{Loop program.}
    \end{subfigure}
    \begin{subfigure}{0.3\textwidth}
        \begin{qprognamed}{alias-unroll-2}
            qubit q0
            qubit q1
            qubit q2
            qubit q3
            qubit q4
            H q0
            ZZ q0,q1
            H q1
            H q1
            ZZ q1,q2
            H q2
            H q2
            ZZ q2,q3
            H q3
            H q3
            ZZ q3,q4
            H q4
        \end{qprognamed}
        \caption{Unrolled circuit.}
    \end{subfigure}
    \caption{Periodic feature in the unrolled loop can be captured.}
    \label{fig:alias_captured}
\end{figure}

To see the second problem, we note that detection of memory aliasing \cite{10.5555/1177220} is usually solved by an Integer Linear Programming (ILP) problem solver such as Z3\cite{10.1007/978-3-540-78800-3_24}. However, a general ILP problem is NP-complete in theory and may take long time to solve in practice. Fortunately, we will see that all problems that we are facing can be solved efficiently in $O(1)$ time without an ILP solver.

We consider two references to a same qubit array:  
$q\left[k_1i+b_1\right],$ $q\left[k_2i+b_2\right],$ $i\in T$, where $T$ is the loop interval when the loop range is known and $\mathbb{Z}$ when unknown.

\begin{definition}\label{def-ilqa}\textbf{In-loop qubit aliasing}:
To check whether two instructions can always be executed together, we have to check if one qubit reference may be an alias of another, that is, $(\exists i \in T)\left(k_1i+b_1=k_2i+b_2\right).$
\end{definition}

This problem can be easily solved by checking whether $(b_2-b_1)$ is a multiple of $(k_1-k_2)$ and $\frac{b_2-b_1}{k_1-k_2}$ lies in $T$.

\begin{definition}\label{def-alqa}
\textbf{Across-loop qubit aliasing}:
To check whether there is an across-loop dependency between two instructions, we have to check if one qubit reference may be an alias of another qubit reference \textbf{several iterations later}. Thus, we need to find the minimal increment $\Delta i\geqslant 1$, s.t.
\begin{equation}
    (\exists i \in T)\left((i+\Delta i \in T)\wedge (k_1i+b_1=k_2(i+\Delta i)+b_2)\right).
\end{equation}
\end{definition}

This issue can be reduced to the Diophantine equation 
\begin{equation}
    (k_2-k_1)i+k_2(\Delta i)=b_1-b_2, i\in T, i+\Delta i \in T, \Delta i\geqslant 1,
\end{equation} which can be solved in $O(1)$ time [\textit{see Appendix \ref{appendix:diophantine-noilp}}]. We solve the equation every time when needed rather than memorizing its solution. A visualization of across-loop qubit aliasing is presented in Figure \ref{fig:across-loop-aliasing}.
\begin{figure}
    \centering
    \includegraphics[max width=0.40\textwidth]{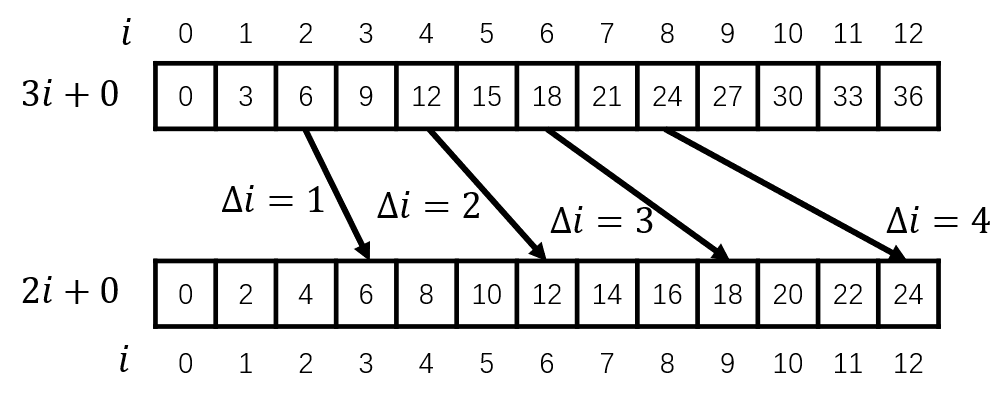}
    \caption{\label{fig:across-loop-aliasing}An example for across-loop qubit aliasing with $k_1=3$ and $k_2=2$. For $T=\mathbb{Z}$, $\Delta i=1$, while for $T=\left[4, 10\right]$, $\Delta i=2$.}
\end{figure}

\subsection{Instruction data dependency}
\label{subsec:data_dep}
One most important step in rescheduling a loop is to find the data dependencies - instrucions that can not be reordered while scheduling. Previous work mostly defined instruction dependency according to matrix commutativity: the order of two instructions can change if their unitary matrices satisfy $AB=BA$. This captures most commutativity between gates, but not all. Here, we relax this requirement by establishing several novel and more relexed commutativity rules between quantum instructions.
Since $CZ$ gates is the only two-qubit gate we use and any two $CZ$ gates commute with each other, what we need to care about is commutativity between $CZ$ gates and one-qubit gates.
\begin{definition}
(CZ conjugation) If for one-qubit gates $U_A$, $U_B$, $V_A$ and $V_B$, we have
$CZU_AU_BCZ=V_AV_B$,
we say $CZ$ conjugates $U_A\otimes U_B$ into $V_A\otimes V_B$.
\end{definition}

Conjugation allows us to swap a $CZ$ gate with a pair of one-qubit gates, at the price of changing $U_A$ and $U_B$ to $V_A$ and $V_B$ correspondingly. 
The following theorem identifies all possible conjugations.
\begin{theorem}\label{rule-1}  
(CZ conjugation of single qubit gates) $CZ$ conjugates $U_A\otimes U_B$ into some $V_A\otimes V_B$ \textbf{if and only if} $U_A$ and $U_B$ are diagonal or anti-diagonal:  $U_i=R_Z(\theta)$ or $U_i=R_Z^+(\theta)$ for $i\in\{A,B\}$.
\end{theorem}

\begin{note}
The antidiagonal rule has been named ``EjectPhasedPaulis'' in \cite{quantum_ai_team_and_collaborators_2020_4062499}. However we propose the rules for both necessity and sufficiency: no more commutation rules can be obtained at gate level.
\end{note}

Since identity matrix $I$ is diagonal, $U_A$ and $U_B$ can be thought of as going under conjugation separately. Thus, we only need to consider two special cases: $I\otimes R_Z$ and $I\otimes R_Z^+$. Note that in conjugation rules $R_Z^+$ will always introduce a $Z$ gate to the other qubit. This inspires us to generalize Theorem \ref{rule-1} for a generalized form of $CZ$ defined in the following:

\begin{definition}
(Generalized $CZ$ gates) For $x,y\in\{0,1\}$, we define following variants of $CZ$ gate:
\begin{equation*}
\begin{aligned}
    &CZ_{11}[a,b]=CZ[a,b],\qquad\ \ CZ_{00}[a,b]=-Z[a]Z[b]CZ[a,b]\\
    &CZ_{10}[a,b]=Z[a]CZ[a,b],\ \ CZ_{01}[a,b]=Z[b]CZ[a,b]
\end{aligned}
\end{equation*}
\end{definition}

Equivalently, $CZ_{xy}$ can be defined as follows: $CZ_{xy}\ket{ab}=(-1)^{\delta_{ax}\delta_{by}}\ket{ab}$, where $\delta_{ij}$ is Kronecker delta. Now we have the following commutativity rules for generalized $CZ$:
\begin{theorem}
\label{th:generalized_conjugation} (Generalized $CZ$ conjugation of single qubit gates)  When exchanged with $R_Z^+$, $CZ$ gate changes into one of its variants by toggling the corresponding bit.
\begin{enumerate}
    \item $R_Z(\alpha)[b]CZ_{xy}[a,b]=CZ_{xy}[a,b]R_Z(\alpha)[b]$;
    \item $R_Z^+(\alpha)[b]CZ_{xy}[a,b]=CZ_{x(1-y)}[a,b]R_Z^+(\alpha)[b]$.
\end{enumerate}
\end{theorem}

Since generalized $CZ$ gates are also diagonal, they commute with each other and can be scheduled just as ordinary $CZ$ gate and converted back to $CZ$ by adding $Z$ gates.

\subsection{Instruction resource constraint}

Qubits have properties that resemble both data and resource: qubits work as quantum data registers and carry quantum data; meanwhile, qubit-level parallelism allows all instructions, if they operate on different qubits, to be executed simultaneously. This results in a surprising property for quantum programs: the resources should be described using linear expressions, instead of by a static ``resource reservation table'' as in the classical case. Using the rules for detecting qubit aliasings, we simply check if there is an aliasing between the qubit references from two instructions, that is, the two instructions share a same qubit at some iteration and cannot be executed simultaneously.


\section{Rescheduling loop body}
\label{sec:resched_loop_body}
Now we are ready to present the main algorithm for pipelining quantum loop programs. It is based on modulo scheduling via hierarchical reduction \cite{allan1995software}, but several modifications to the original algorithm are required to fit into scheduling quantum instructions on qubits. The entire flow of our approach is depicted in Figure \ref{fig:compilation-flow}.
\begin{figure}
    \includegraphics[max width=0.5\textwidth]{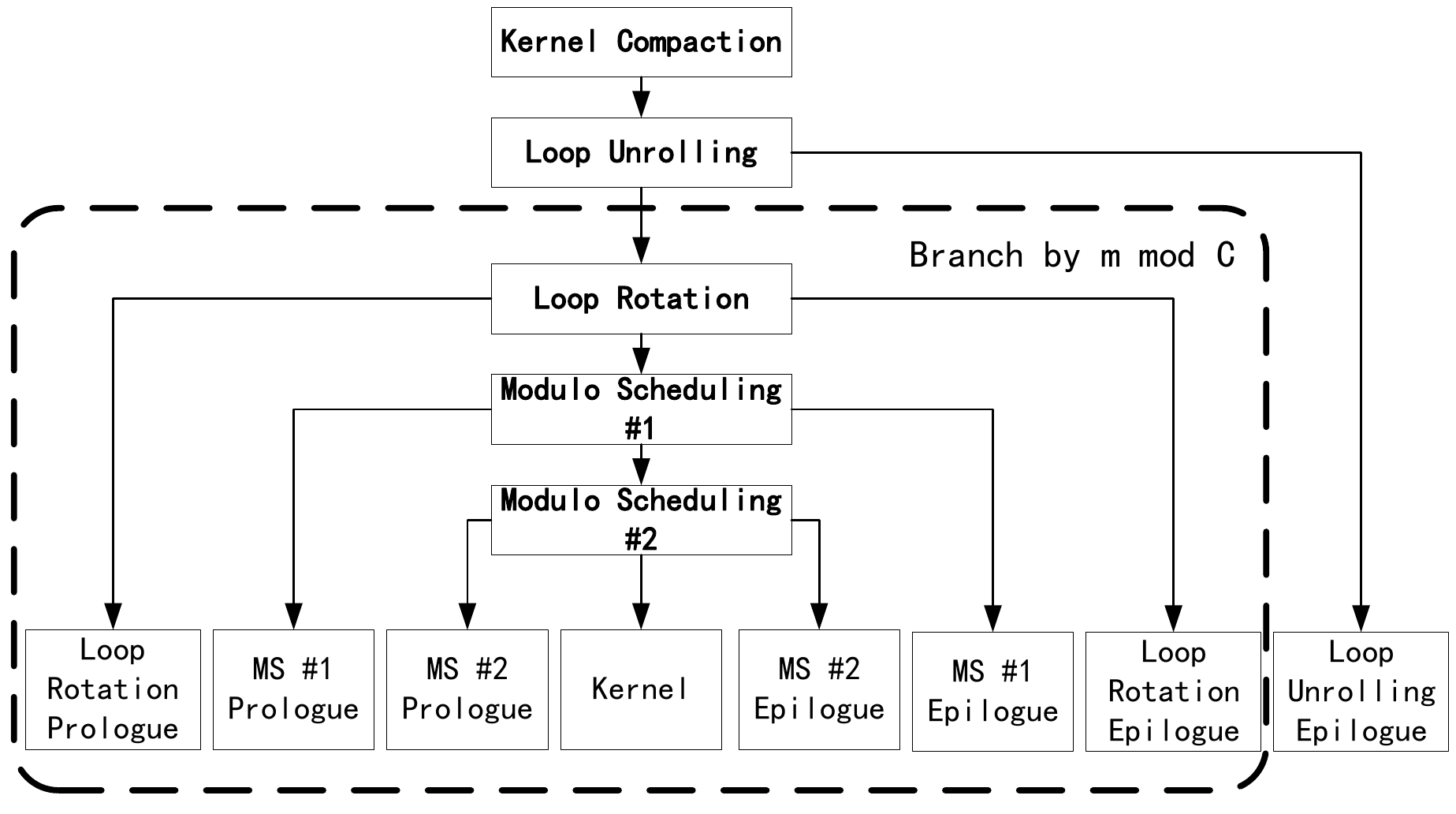}
    \caption{\label{fig:compilation-flow}The entire compilation flow of our approach.}
\end{figure}
For simplicity we suppose the number of iterations is large enough so that we don't worry about generating a long prologue/epilogue.

\subsection{Loop body compaction}

At first we compact the loop kernel to merge the gates that can be trivially merged, including: (a) adjacent single qubit gates; 
(b) diagonal or antidiagonal single qubit gates and their nearby single qubit gates, maybe at the other side of a $CZ$ gate; and (c) adjacent $CZ$ gates. 
To this end, we define the following compaction procedure, which considers the potential aliasing between qubits:
\begin{definition}
    A greedy procedure for compacting loop kernel:
    \begin{itemize}
        \item Initialize all qubits with an ideneity gate.
        \item Place all instructions one by one. Initialize \textbf{operation} to ``Blocked''. Check the new instruction (A) against all placed instructions (B). Update \textbf{operation} according to Table \ref{tab:compaction}.
        \item Perform the last \textbf{operation} according to the table.
        \begin{itemize}
            \item ``Blocked'' means the instruction is put at the end of the instruction list.
            \item ``Merge with B'' means the single qubit instruction is merged with the placed single qubit gate B. If the placed gate is an antidiagonal, $Z$ gates should be added for uncancelled $CZ$ gates that occur earlier but are placed after the antidiagonal.
            \item ``Cancelled'' means two $CZ$ gates are cancelled. Note that the added $Z$ gates are not cancelled. Also, a third arriving $CZ$ can ``uncancel'' a cancelled $CZ$, which we also call as ``Cancelled''.
        \end{itemize}
    \end{itemize}
\end{definition}
\begin{figure*}
    \begin{tabular}{|
    >{\columncolor[HTML]{EFEFEF}}l |l|l|l|l|}
    \hline
    \cellcolor[HTML]{C0C0C0}A\textbackslash{}B & \cellcolor[HTML]{EFEFEF}SQ with same qubit & \cellcolor[HTML]{EFEFEF}SQ with in-loop aliasing & \cellcolor[HTML]{EFEFEF}CZ with same qubit & \cellcolor[HTML]{EFEFEF}CZ with aliasing qubit \\ \hline
    Diagonal SQ                                & Merge with B                               & Blocked                                          &                                            &                                                \\ \hline
    AntiDiagonal SQ                            & Merge with B                               & Blocked                                          &                                            & Blocked                                        \\ \hline
    General SQ                                 & Blocked                                    & Blocked                                          & Blocked                                    & Blocked                                        \\ \hline
    CZ                                         & Blocked                                    & Blocked                                          & If exactly-same then Cancel                &                                                \\ \hline
    \end{tabular}
    
    \captionof{table}{\label{tab:compaction}Operation table for loop kernel compaction. Empty cell means using previous operation. Check is performed from left to right, so antidiagonal can pass through $CZ$ with a same qubit and an aliasing qubit.}
\end{figure*}

This compaction can be done in two directions: compacting to the left or to the right. They can be seen as the results of ASAP schedule and ALAP  correspondingly.
\begin{figure}
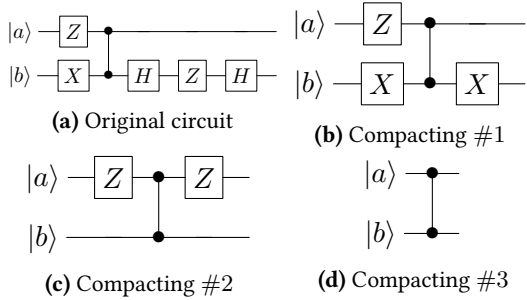

    \centering%
    \begin{subfigure}{0.2\textwidth}
        \centering%
        \begin{qprognamed}{greedy_leftcompact_1}
            qubit a
            qubit b
            Z a
            X b
            ZZ a,b
            H b
            Z b
            H b
        \end{qprognamed}
      \caption{Original circuit}
    \end{subfigure}%
    \begin{subfigure}{0.2\textwidth}
        \centering%
        \begin{qprognamed}{greedy_leftcompact_2}
            qubit a
            qubit b
            Z a
            X b
            ZZ a,b
            X b
        \end{qprognamed}
      \caption{Compacting $\#1$}
    \end{subfigure}
    \begin{subfigure}{0.2\textwidth}
        \centering%
        \begin{qprognamed}{greedy_leftcompact_3}
            qubit a
            qubit b
            Z a
            ZZ a,b
            Z a
        \end{qprognamed}
      \caption{Compacting $\#2$}
    \end{subfigure}
    \begin{subfigure}{0.2\textwidth}
        \centering%
        \begin{qprognamed}{greedy_leftcompact_4}
            qubit a
            qubit b
            ZZ a,b
        \end{qprognamed}
      \caption{Compacting $\#3$}
    \end{subfigure}
    \caption[Compaction converges only after multiple applications.]{Compacting more than once yields better result.}
    \label{fig:greedy-leftnormalize-bad}
  \end{figure}
However, this procedure does not guarantee compacting once will converge: not all the outputs from the procedure are fixpoints of the procedure. For example, the circuit in Figure \ref{fig:greedy-leftnormalize-bad}  only converges after three applications of left compaction. In general, we have the following:
\begin{theorem}\label{thm:compaction}
    Compacting three times results in a fixpoint of the compaction procedure.
\end{theorem}

Note that we allow using unknown single-qubit gates. If all components are known to be diagonal or antidiagonal, the product of these matrices is also diagonal or antidiagonal [\textit{see Appendix  \ref{appendix:compaction}}]. Otherwise, we can only see the product as a general matrix. However, this does not affect our result of three-time compaction.
\begin{figure}
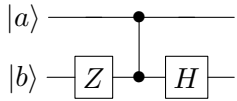

    \centering
    \begin{qprognamed}{bidir_compact}
    qubit a
    qubit b
    Z b
    ZZ a,b
    H b
    \end{qprognamed}
    \caption{Left compaction will miss the chance of compacting the $Z$ gate and the $H$ gate.}
    \label{fig:bidirectional-merge}
\end{figure} 
Also compacting in one direction does not capture all chances of merging. Figure \ref{fig:bidirectional-merge} shows that some single-qubit merging changes are missed out. In practice we perform a left compaction after a right compaction.

{
\subsubsection{Loop unrolling and rotation}
Loop kernel compaction can only discover gate merging and cancellation in one iteration. However, gate merging and cancellation can also occur across iterations. For example, in Figure \ref{fig:alias_captured} the last $H$ gate in the previous iteration can be merged and cancelled with the first $H$ gate in the next iteration. This kind of cancellation cannot be discovered by software pipelining either, since it is a reordering technique and cannot cancel instructions out.

An instruction $i$ in one iteration may merge or cancel with instruction $j$ from $t\geqslant 1$ iterations later. All potential merging of single qubit gates and cancellable $CZ$ gates can be written out by enumerating all pairs of instructions. 
Loop rotation\cite{loop_rotation} is an optimization technique to convert across-loop dependency to in-loop dependency (so that some variables can be privatized and optimized out). Consider a loop ranging from $m$ to $n$:
    $\left\{A_iB_iC_i\right\}_m^n.$
Here, $A_i$ can be rotated to the tail of the loop:
$   A_m\left\{B_iC_iA_{i+1}\right\}_m^{n-1}B_nC_n,$
and $C_i$ and $A_{i+1}$ are now in one iteration. If $C_i$ writes into a temporary variable and $A_{i+1}$ reads from it, this variable can be privatized. 
For merging candidates with $t=1$, we can use a similar procedure:
\begin{definition}
    An instruction is considered movable if it satisfies one of following conditions:
    \begin{itemize}
        \item The instruction is a single-qubit gate, and there are no gates on the same qubit or on an aliasing qubit before it; in this case the instruction can be rotated to the right.
        \item The instruction is a $CZ$ gate, and there are no \textbf{single-qubit gates} on the same qubit or on ailasing qubits; in this case the instruction can be rotated to the right.
        \item The instruction is a $CZ$ gate, and there are no \textbf{single-qubit gates} on the same qubit or on ailasing qubits \textbf{except} the $CZ$ gate has only one linear offset reference with $k=0$ and there is a single-qubit gate on this qubit. In this case, the instruction will be rotated to the right along with this single qubit gate.
    \end{itemize}
\end{definition}

This definition of movable instructions guarantees the programs before and after the rotation are equivalent.
We use the following procedure to rotate one instruction from left to  right:
\begin{enumerate}
    \item Find the first \textbf{unmarked} movable instruction that, there exists another instruction to merge or cancel with $t=1$.
    \item \textbf{Mark} the chosen instruction, and rotate the instruction to the right. The instruction is added to prologue and the others added to epilogue.
    \item Perform left compaction on the new loop kernel. Note that the left-compaction algorithm is modified, so that merging single-qubit gates or cancelling $CZ$ gates will clear the mark.
    \item If there is no rotatable instruction, stop the procedure.
\end{enumerate}


\begin{corollary}
    If the original loop has only candidates with $t=1$ and no one-qubit gate merges with itself, this procedure eliminates all across-loop merging or cancellation. That is, if we unroll the loop after rotation, the unrolled quantum ``circuit'' should be a fixpoint of compaction procedure.
\end{corollary}

However, loop rotation can only handle potential gate merging across one iteraion (i.e. from nearby iterations). To handle potential merging across many iteraions, we adopt loop unrolling from classical loop optimization. While the major objective for loop unrolling is usually to reduce branch delay, Aiken et al. \cite{aiken1988optimal} also used loop unrolling to unroll first few iterations of loop and schedule them ASAP, so that repeating patterns can be recognized into an optimal software pipelining schedule.
Our approach uses modulo scheduling instead of kernel recognition, but we can still exploit the power of loop unrolling to capture patterns that require many iterations to reveal. The key point is that unrolling decreases $t$. Suppose we use a graph to represent all ``candidates for instruction merging'', with edge $A\stackrel{t}{\longrightarrow} B$ indicating instruction $A$ will merge with or cancel out instruction $B$ from $t$ iterations later, if we unroll the loop by $C$ times, the weight of the edges in the graph will decrease.

\begin{example}
    Figure \ref{fig:unroll_perfect_pipelining} gives an example showing the connection between the ``merging graph'' before unrolling and the one after unrolling: if 
$        \forall t, C\geqslant t$, there are no edges with $t>1$.
\end{example}

\begin{figure}
    \begin{center}
        \begin{subfigure}{0.45\textwidth}
            \begin{center}
            \includegraphics[max width=\textwidth, max height=0.5\textheight]{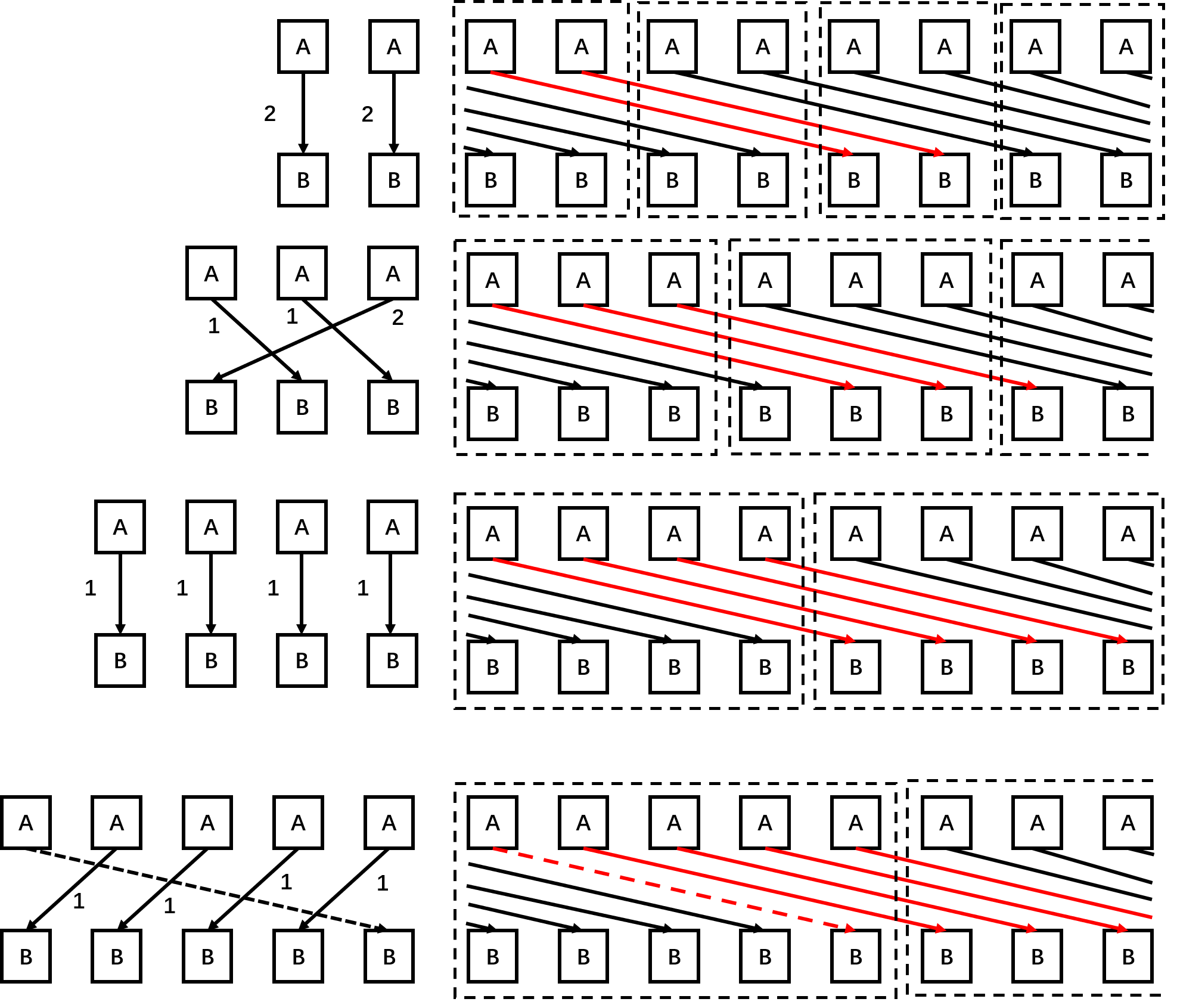} 
            \end{center}
        \end{subfigure}
    \end{center}
    \caption{Example for the QDGs of loop ``$A\stackrel{4}{\longrightarrow} B$'' unrolled 2, 3, 4 and 5 times. Unrolling the loop decreases the edge weight $t$. When $C=max\left\{t\right\}$ all edges will be decreased to weight 1.}
    \label{fig:unroll_perfect_pipelining}
\end{figure}

There is a tradeoff between generated code length (determined by $C$) and remaining $t>1$ edges. For example, if there is an edge with $t=10000$, we are not likely to unroll the loop for $10000$ times just to merge the two single qubit gates. Also for eliminating self-cancelling $CZ$ gates (i.e. $CZ$ gates on a pair of constant qubits), we may want $C\geqslant 2$ and $C$ even.
In the following discussion we use $C$ as a configurable variable in our algorithm determining the maximal allowed unroll time (and the minimal time of iterations of the loop).
The new unrolled loop will be in the form
\begin{equation}
\begin{split}
    &for(i=m;i\leqslant n;i+=C)\left\{op(k_ti+b_t)\right\}\\
    &for(i=m';i\leqslant n;i+1)\left\{op(k_ti+b_t)\right\}
\end{split}
\end{equation}
and the first loop should be written into
\begin{equation}
    for(i=0;i\leqslant n';i+=1)\left\{op(Ck_ti+b_t+mk_t)\right\}
\end{equation}
where $n'=\floor*{\frac{n-m+1}{C}}-1$ and $m'=C(n'+1)+m$.
This step of transformation makes sure the loop stride is still $1$ after loop unrolling.
Note that item $(mk_t)$ appears in every offset of the loop body. If $m$ is unknown we can't proceed with our algorithm.
Fortunately, since $m=pC+q, q=m\mod C$, we have
    $Ck_ti+b_t+mk_t=Ck_t(i+p)+b_t+qk_t,$
showing that when the range is unknown, the results of array dependency depend only on the Euclidean modulo $q=m \mod C$. In this case, we can generate $C$ copies of code for each case of $q$, and perform following parts of the algorithm on each copy.

Let us briefly summarize our compilation flow till now: we compact the loop kernel, unroll the loop by $C$, and rotate some instructions in the  unrolled loop kernel. The unrolling step may copy the loop by $C$ times, and steps after unrolling (including rotation) will be performed on each copy.
}

\subsection{Modulo scheduling}
Our next step is modulo scheduling borrowed from~\cite{lam1988software}: 
\begin{enumerate}
    \item Find in-loop and loop-carried dependencies.
    \item Estimate an initialization interval $II$. For simplicity we use binary search and the maximum $II$ is total instruction count. Use Floyd to check validity.
    \item Using Tarjan algorithm to find strong connected components and schedule all SCCs by in-loop dependency subgraph.
    \item Merge every SCC in DDG into one node, obtaining a new DDG.
    \item Schedule the new DDG by list scheduling.
\end{enumerate}
There are some major differences between quantum programs and the classical programs considered in ~\cite{lam1988software}:
\subsubsection{Quantum dependency graph}

The instruction dependency for quantum programs is described by a QDG (Quantum Dependency Graph) as a generalization of DDG (Data Dependency Graph), where vertices represent instructions and edges represent precedence constraints that must be satisfies while reordering. In modulo scheduling, a dependency edge is described by two integers: $min$ and $dif$. Suppose there is an edge pointing from instruction $A$ to instruction $B$ with parameter $\left(min, dif\right)$, it means ``instruction $B$ from $dif$ iterations later should be scheduled at least $min$ ticks later than instruction $A$ in \textbf{this} iteration''. Recall from Section \ref{subsec:qubit_alias} and \ref{subsec:data_dep}, our dependency is defined by the rules:
\begin{enumerate}
    \item There are no dependencies between $CZ$ gates, or between a $CZ$ and a diagonal single qubit gate.
    \item In-loop dependency: if two offsets are on the same qubit array and reveal in-loop qubit aliasing, there is a dependency edge $(1, 0)$ between the corresponding instructions. To unify with across-loop, we set $\Delta i=0$.
    \item Across-loop dependency: if two offsets are on the same qubit array and reveal across-loop qubit aliasing with $\Delta i$, there is a dependency edge $(1, \Delta i)$ between the corresponding instructions.
    \item Exception on antidiagonal gates: if the qubit $(k_1i+b_1)$ of an antidiagonal gate aliases with one operand $k_2i+b_2$ of a $CZ$ gate and $k_1=k_2$, we remove the edge if there's no aliasing on the other operand.
    \item Exception on single qubit gates: if two single qubit gates operate on the same qubit array where offsets $(k_1i+b_1)$ and $(k_2i+b_2)$ aliases with each other and $k_1=k_2$, we specify the dependency edge to be valued $(0, \Delta i)$, that is, $min=0$ rather than $min=1$.
\end{enumerate}
There may be multiple edges in the graph connecting the same pair of instructions; for example, an in-loop dependency and an across-loop dependency between the two instructions. Since we are going to use Floyd algorithm on the graph to compute largest distance in modulo scheduling, we only need the edge with the maximal $(min-II\cdot dif)$ after assigning $II$. Fortunately we don't need to save all multiple edges, since the following theorem guarantees that we can compare $(min-II\cdot dif)$ before assigning different $II$s.

\begin{theorem}
\label{th:remove-multiple-edge}
    Suppose $(min_1, dif_1), (min_2,dif_2)$ are two edges with $min_1\leqslant 1$, $min_2\leqslant 1$ and $dif_1>dif_2$. Then for all $II\geqslant 1$, we have: 
        $min_1-II\cdot dif_1 \leqslant min_2-II\cdot dif_2.$
 \end{theorem}

This theorem allows us to sort multiple edges by lexical ordering of $(dif, -min)$ (i.e. compare $dif$ first, and compare $(-min)$ if $dif_1=dif_2$) and the smallest one is exactly the edge with maximal $(min-II\cdot dif)$.

\begin{figure}
    \begin{subfigure}{0.4\textwidth}
        \begin{lstlisting}[frame=single]
for x=m to n do
    CNOT q1[x-50],q0[x+0];
    CNOT q1[x-50],q0[x+0];
end for
        \end{lstlisting}
        \caption{Loop program.}
    \end{subfigure}
    \begin{subfigure}{0.3\textwidth}
        \includegraphics[max width=\textwidth]{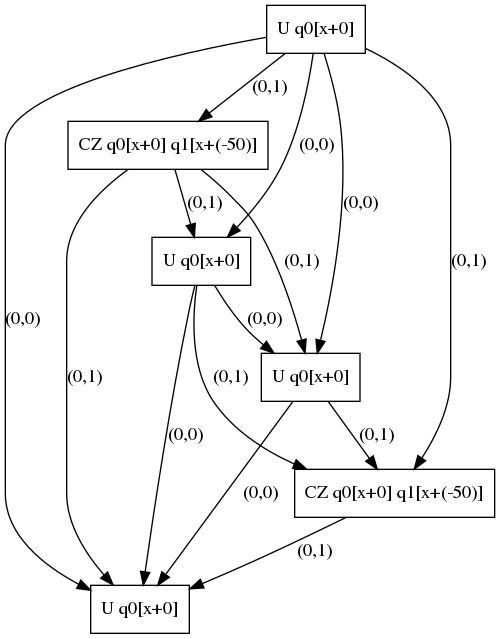}
        \caption{Corresponding QDG.}
    \end{subfigure}
    \caption{Quantum dependency graph example. Tuples represent $(min, dif)$.}
    \label{fig:qdg_example}
\end{figure}

\subsubsection{Resource conflict handling}
Another important issue when inserting an instruction into modulo scheduling table or merging two strong connected components is resource conflict: there is no dependency between two $CZ$ gates, yet they may not be executed together because they may share a same qubit.
To solve this issue, let us first introduce several notations:
    \begin{enumerate}
        \item $II$ is the current iteration interval being tested.
        \item $L$ is the length of the original loop kernel.
        \item The $c$-th instruction in the original loop is placed in the modulo scheduling table at tick $t=pII+q$, where $p\geqslant 0$, $0\leqslant q<II$.
    \end{enumerate}

\begin{example}
    Figure \ref{fig:ms_example} is a simple example for modulo scheduling. In this case, $II=2$ and $L=4$. Instructions are placed at time slot $0,2,3,4$. Thus, $A$ from one iteration, $B$ from a previous iteration, and $D$ from previous 2 iterations are executed simultaneously, while $C$ is executed alone.
\end{example}

We use the retrying scheme: if a resource conflict is detected, try next tick. The basic approach to detect resource conflict is detecting in-loop qubit aliasing. This leads to two new problems that do not exist in the   classical case:
\begin{enumerate}
    \item The array offsets of instruction operands may increase. As $t$ increases, $p$ also increase, and the instruction comes from one more iteration earlier, thus changing array offsets.
    \item The pair of instructions for resource conflict checking may not both exist in some iterations. Increasing $t$ leads to a long prologue and long epilogue, shrinking the range for loop kernel, and may eliminate the resource conflict that once existed (when the loop range is known).
\end{enumerate}

\begin{example}
    Suppose when generating the schedule in Figure \ref{fig:ms_example}, we have inserted instructions $A$, $B$ and $C$, and are ready to insert $D$ at time slot $4$.
    \begin{enumerate}
        \item Since $4=2II+0$, the $D$ in the loop kernel is from two iterations earlier compared with the iteration that the $A$ is in. We have to decrease offset of $D$ operands by $2i$. The offseted index may no longer conflict with $A$.
        \item When checking if there is resource conflict between $D$ and $A$, we only need to check the case where both iterations are valid; that is, $i=2$. 
        This means the scheduling is still valid even if $A_0$ has a resource conflict with $D_{-2}$, since $D_{-2}$ does not even exist.
    \end{enumerate}
\end{example}
\begin{figure}
    \begin{subfigure}{0.15\textwidth}
        \begin{center}
        \includegraphics[max width=\textwidth]{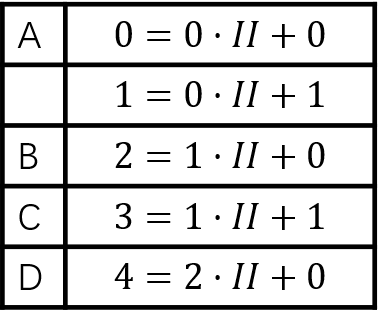}
        \end{center}
        \caption{Rescheduled single iteration. $II$=2.}
    \end{subfigure}
    \hspace{0.1cm}
    \begin{subfigure}{0.1\textwidth}
        \begin{center}
        \includegraphics[max width=\textwidth]{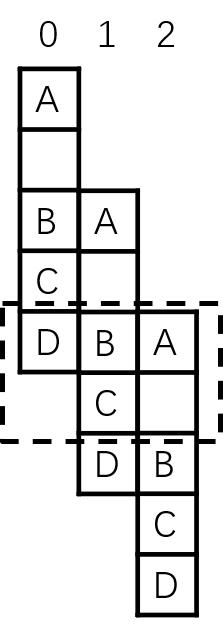}
        \end{center}
        \caption{Issuing each iteration reveals loop kernel.}
    \end{subfigure}
    \hspace{0.1cm}
    \begin{subfigure}{0.15\textwidth}
        \begin{center}
        \includegraphics[max width=\textwidth]{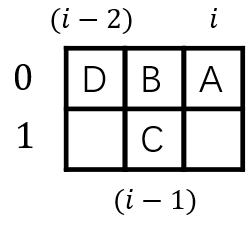}
    \end{center}
        \caption{Modulo scheduling table. Column index represents original iteration.}
    \end{subfigure}
    \caption{\label{fig:ms_example}Example for modulo scheduling loop $A_iB_iC_iD_i$. In this case $II=2$, $L=4$, $T=\left[0,2\right]$.}
\end{figure}

In the original modulo scheduling and other classical scheduling algorithms, the retry strategy only allows $II$ retries. For example, if there is not enough $ALU$ or $FPU$ for instruction $A_i$ in modulo scheduling table tick $q$, there is also not enough resource for instruction $A_{i-1}$ from previous iteration. However, this is not true for our case, and we have to modify the strategy.

\begin{example}
    Suppose we perform modulo scheduling on the program in Figure \ref{fig:resource_move}. Since the three $CZ$s are exactly the same, we may expect $II=3$ due to resource conflict. However, if we allow more retries, these $CZ$s can be separated into different iterations and can be executed concurrently with $CZ$s from other iterations.
\end{example}

    \begin{figure*}
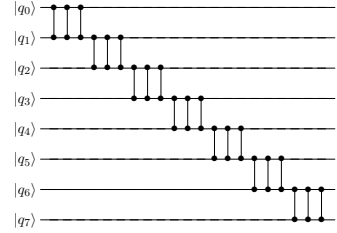
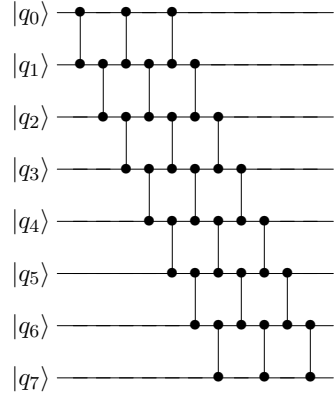

    \begin{center}
        \begin{subfigure}{0.24\textwidth}
            \begin{lstlisting}[frame=single]
for x=0 to 6 do
    CZ q[x],q[x+1];
    CZ q[x],q[x+1];
    CZ q[x],q[x+1];
end for
                        \end{lstlisting}
        \caption{Original Program.}
        \end{subfigure}
        \hspace*{1cm}
        \begin{subfigure}{0.24\textwidth}
\begin{qprognamed}{resmove_1}
    qubit q0
    qubit q1
    qubit q2
    qubit q3
    qubit q4
    qubit q5
    qubit q6
    qubit q7
    ZZ q0,q1
    ZZ q0,q1
    ZZ q0,q1
    ZZ q1,q2
    ZZ q1,q2
    ZZ q1,q2
    ZZ q2,q3
    ZZ q2,q3
    ZZ q2,q3
    ZZ q3,q4
    ZZ q3,q4
    ZZ q3,q4
    ZZ q4,q5
    ZZ q4,q5
    ZZ q4,q5
    ZZ q5,q6
    ZZ q5,q6
    ZZ q5,q6
    ZZ q6,q7
    ZZ q6,q7
    ZZ q6,q7
\end{qprognamed}
        \caption{Unrolled Program, for a clearer view.}
        \end{subfigure}
\end{center}
\begin{center}
        \begin{subfigure}{0.32\textwidth}
\begin{lstlisting}[frame=single]
CZ q[0],q[1];
CZ q[1],q[2];
CZ q[0],q[1]; CZ q[2],q[3];
CZ q[1],q[2]; CZ q[3],q[4];
for x=4 to 6 parallel do
    CZ q[x-4],q[x-3];
    CZ q[x-2],q[x-1];
    CZ q[x],q[x+1];
end for
CZ q[3],q[4]; CZ q[5],q[6];
CZ q[4],q[5]; CZ q[6],q[7];
CZ q[5],q[6];
CZ q[6],q[7];
\end{lstlisting}
    \caption{Software pipelined version.}
\end{subfigure}
\hspace*{1cm}
\begin{subfigure}{0.24\textwidth}
    \begin{center}
    \begin{qprognamed}{resmove_2}
        qubit q0
        qubit q1
        qubit q2
        qubit q3
        qubit q4
        qubit q5
        qubit q6
        qubit q7
        ZZ q0,q1
        ZZ q1,q2
        ZZ q0,q1
        ZZ q2,q3
        ZZ q1,q2
        ZZ q3,q4
        ZZ q0,q1
        ZZ q2,q3
        ZZ q4,q5
        ZZ q1,q2
        ZZ q3,q4
        ZZ q5,q6
        ZZ q2,q3
        ZZ q4,q5
        ZZ q6,q7
        ZZ q3,q4
        ZZ q5,q6
        ZZ q4,q5
        ZZ q6,q7
        ZZ q5,q6
        ZZ q6,q7
    \end{qprognamed}
\end{center}
    \caption{Software pipelined version, unrolled.}
\end{subfigure}
\end{center}
        \caption{\label{fig:resource_move}Three $CZ$ gates in a row. Although there seems to be resource conflicts, the minimal $II=1$.}
    \end{figure*}

We consider the general case where loop range is unknown. When placing an instruction in the modulo scheduling table, we check its operands with all operands scheduled at this tick. Suppose now we check operand $(k_2(i-p_2)+b_2)$ with operand $(k_1(i-p_1)+b_1)$, and we find an aliasing, that is, 
    $\exists i_0\in \mathbb{Z}, k_2(i_0-p_2)+b_2=k_1(i_0-p_1)+b_1.$ 
In case $k_1=k_2$, 
        $\forall i\in \mathbb{Z}, k_2(i-p_2)+b_2=k_1(i-p_1)+b_1.$
    When $k_1=0$, this is the same as classical resource scheduling; otherwise,
        $\forall \Delta p\neq 0, \forall i\in \mathbb{Z}, k_2(i-p_2-\Delta p)+b_2\neq k_1(i-p_1)+b_1.$
    This means if we delay the instruction by $\Delta pII$ ticks, the conflict will be resolved. We call it \textbf{false conflict}.
    In case $k_1\neq k_2$, after $\Delta pII$ ticks it will fall in the same time slot. There is still a conflict iff  
        $\exists i_1\in \mathbb{Z}, k_2(i_1-p_2-\Delta p)+b_2=k_1(i_1-p_1)+b_1$; that is, 
        $i_1=i_0+\frac{\Delta pk_2}{k_2-k_1},$ 
    which means
        $(k_2-k_1) | \Delta pk_2.$
    The conflict appears periodically as $\Delta p$ increases. However, in the worst case where 
        $(k_2-k_1) | k_2,$ 
    there is always a conflict and can be seen as classical resource scheduling. We call it, together with the case where $k_1=k_2=0$, \textbf{true conflict}.

We insert an instruction or an entire schedule into the modulo scheduling table in the following way: if there is no conflict, we insert the instructions; if there is only false conflict, we try next tick. As an exception, false conflicts between two single qubit gates are also seen as no conflict; and if there is true conflict, we start a ``death countdown'' before trying next tick: if next $(II-1)$ retries do not succeed, give up, as we do in classical retry scheme.

\subsubsection{Inversion pair correction}
The commutativity between antidiagonal $R_Z^+$ gates and $CZ$ gates comes at a price of a Z gate. In modulo scheduling stage we allowed them to commute freely, ignoring the generated Z gates. Now we have to fill them back to ensure equivalence. By the term ``inversion'', we mean that our scheduling alters the execution order of instructions \textbf{compared with original ordering}:
\begin{definition}
    If the original $c$th instruction is modulo-scheduled at $t=pII+q$ in new loop (where the $k$th original loop is issued), we define the absolute order of the instruction to be
        $T=(k-p)L+c=kL+(c-pL).$
\end{definition}
\begin{example}
    Suppose $L=4$ and $B$ in Figure \ref{fig:ms_example} is the second instruction in the original loop ($c=1$). $B$ is placed in the modulo scheduling table at $p=1$ and $q=0$. 
    \begin{enumerate}
        \item The first $B$ instruction is issued in the prologue (incomplete loop kernel) where the second $(k=1)$ iteration is issued. Thus the absolute order of the instruction is $T=1$.
        \item The second $B$ instruction is issued in the loop kernel where the third $(k=2)$ iteration is issued. Thus the absolute order is $T=5$.
        \item The third $B$ instruction is issued in the epilogue (again incomplete loop kernel) where the fourth $(k=3)$ iteration is issued (or, should be issued). The absolute order is $T=9$.
    \end{enumerate}
    We see that the absolute order is exactly the time when the instruction is executed in the original loop.
\end{example}

Our idea is to check all inversion pairs in the modulo schedule. There are two kind of order-inversions:
\begin{definition}
\label{def:inversion}
\begin{enumerate}
    \item In-loop inversion: For two instructions in the $k$-iteration in new scheduling (i.e. the iteration where $k$th iteration of original loop is issued), if the first precedes the second while its absolute order succeeds the absolute order of the second instruction: 
        $kL+(c_1-p_1L)>kL+(c_2-p_2L),$
     there is an in-loop inversion.
    \item Loop-carried inversion: For two instructions in $k$-iteration and $(k+r)$-iteration ($r\geqslant 1$), if  
        $kL+(c_1-p_1L)>(k+r)L+(c_2-p_2L),$
     there is an across-loop inversion.
\end{enumerate}
\end{definition}

Since the $kL$ term can be cancelled, inversion pairs in modulo schedule also reveals periodicity. Figure \ref{fig:cross_loop_inversion} shows an example with periodic $r=1$ inversions, and $r=2$ inversions.
Since the term $(k+r)L+(c_2-p_2L)$ increases as $r$ increases, there exists $r_0$ s.t. $\forall r>r_0$ there is no across-loop inversion. We can increase $r$ and find pairs of inversion from iteration $k$ and $(k+r)$, until there is no inversion pair. When finding all inversion pairs, we can check the pairs to see if one is $CZ$ and the other is antidiagonal on one of $CZ$'s operand. If so, we add a $Z$ gate at the tick where $CZ$ is placed.

\begin{figure}
    \begin{center}
    \includegraphics[max width=0.3\textwidth, max height=0.5\textheight]{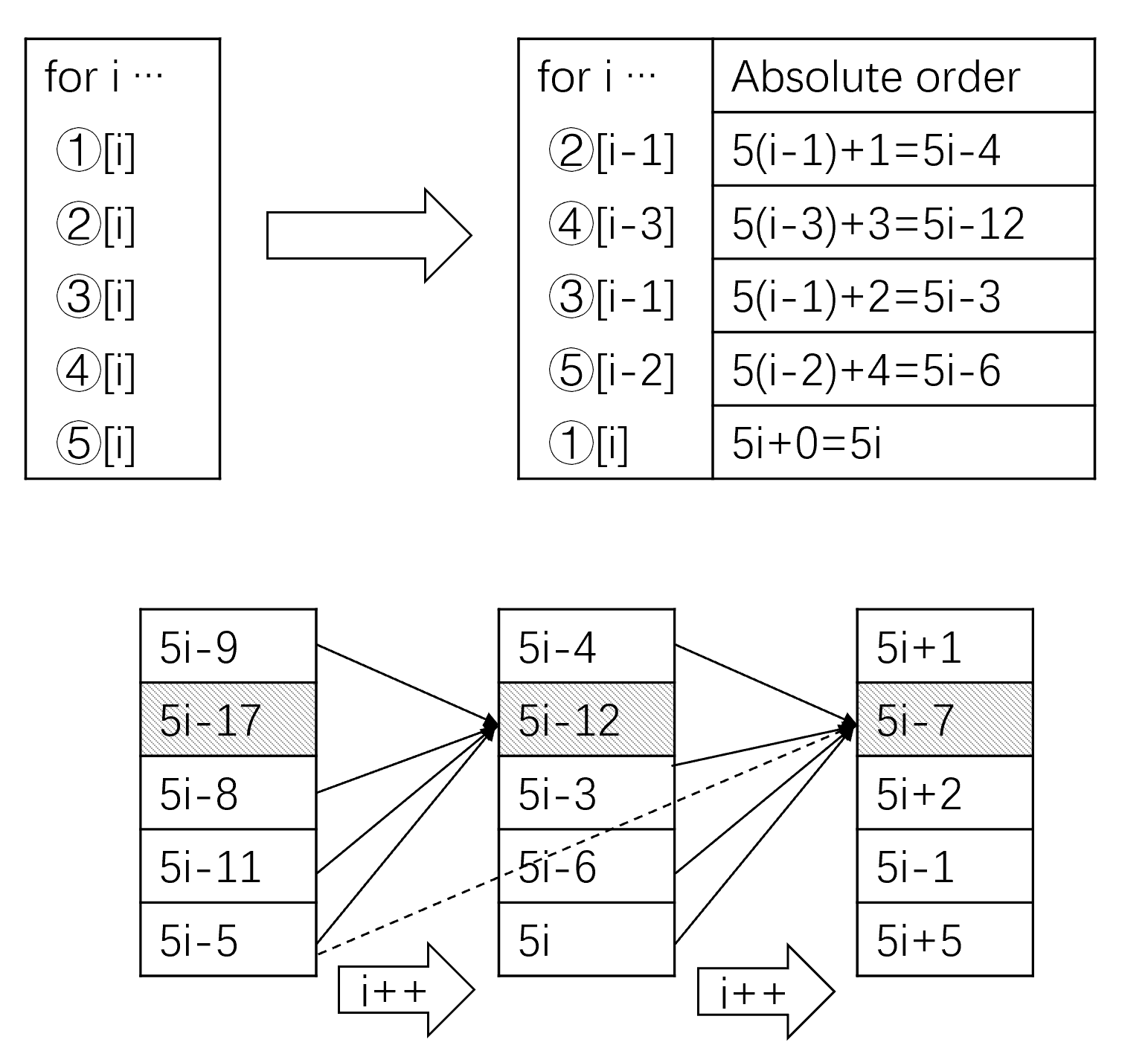} 
    \end{center}
    \caption{An example of inverted pairs of instructions across loop iterations.}
    \label{fig:cross_loop_inversion}
\end{figure}

\subsubsection{Code generation for kernel, prologue and epilogue}

We generate prologue and epilogue by removing non-existing instructions from the loop kernel.

\begin{example}
    Consider in Figure \ref{fig:ms_example} (remember $T=\left[0,2\right]$), the iteration where $k$th original iteration is issued (or should be issued) by enumerating $k$ from $-\infty$ to $\infty$:
    \begin{enumerate}
        \item For $k<0$, $\left\{k, k-1, k-2\right\}\cap T=\Phi$, no instruction is put.
        \item For $k=0$, $\left\{k, k-1, k-2\right\}\cap T=\left\{k\right\}$, only $A$ is put.
        \item For $k=1$, $\left\{k, k-1, k-2\right\}\cap T=\left\{k, k-1\right\}$, $A,B,C$ are put.
        \item For $k=2$, $\left\{k, k-1, k-2\right\}\cap T=\left\{k, k-1, k-2\right\}$. This is the complete loop kernel.
        \item For $k=3$, $\left\{k, k-1, k-2\right\}\cap T=\left\{k-1, k-2\right\}$, $B,C,D$ are put.
        \item For $k=4$, $\left\{k, k-1, k-2\right\}\cap T=\left\{k-2\right\}$, $D$ is put.
        \item For $k>4$, $\left\{k, k-1, k-2\right\}\cap T=\Phi$, no instruction is put.
    \end{enumerate}
\end{example}

For prologue and epilogue, we have to remove instructions from iterations that do not exist; for extra $Z$ gates from the inversion of a $CZ$ and an antidiagonal, removing either gate will make the $Z$ gate disappear. After removing non-existing instructions, we perform compaction and ASAP schedule on the two parts.

For loop kernel, we need to merge the single qubit gates on the same qubit in the same time slot (from the resource conflict exception) by their absolute order.
\begin{figure*}
\begin{center}
\begin{tabular}{|l|l|l|l|l|l|l|l|l|l|l|l|l|}
\hline
\multirow{2}{*}{Test case}            & \multicolumn{3}{c|}{Input Loop}    & \multicolumn{3}{c|}{Output Loop} & \multicolumn{6}{c|}{Known range   results}                                                                \\ \cline{2-13} 
   & ASAP & $C$ & $C$-ASAP & Pre   & K  & Post  & \#Iter & \textbf{K-ASAP} & \textbf{Unroll} & \textbf{Cirq} & QSP\#Iter & \textbf{QSP} \\ \hline
Cluster & 4  & 2 & 5  & 4  & 1  & 4  & 200 & 800  & 203 & {203} & 96 & 104 \\ \hline
Array 1 & 5  & 2 & 10 & 8  & 4  & 5  & 100 & 500  & 500 & {500} & 48 & 205 \\ \hline
Array 2 & 3  & 2 & 5  & 4  & 1  & 4  & 100 & 300  & 201 & {201} & 46 & 54  \\ \hline
Array 3 & 11 & 2 & 17 & 12 & 12 & 17 & 100 & 1100 & 605 & {606} & 48 & 605 \\ \hline
Grover 1      & 13   & 2   & 26                    & 26         & 24      & 871       & 99         & 1287                 & 1287            & {1288} & 15         & 1257         \\ \hline
Grover 2                   & \multicolumn{1}{l|}{71}   & \multicolumn{1}{l|}{2}   & \multicolumn{1}{l|}{141}                   & \multicolumn{1}{l|}{141}      & \multicolumn{1}{l|}{135}    & \multicolumn{1}{l|}{40881}    & \multicolumn{1}{l|}{1000}       & \multicolumn{1}{l|}{71000}                & \multicolumn{1}{l|}{70001}           & \multicolumn{1}{l|}{{71001}} & \multicolumn{1}{l|}{207}        & \multicolumn{1}{l|}{68967}        \\ \hline
QAOA-Hard 1 & 21   & 2   & 41                    & 41         & 40      & 2021      & 1001       & 21021                & 20021           & {20021} & 449        & 20022        \\ \hline
QAOA-Hard 2 & 21   & 2   & 41                    & 41         & 40      & 2061      & 1001       & 21021                & 20021           & {20021} & 448        & 20022        \\ \hline
QAOA-Hard 3 & 16   & 2   & 27                    & 41         & 18      & 1121      & 1001       & 16016                & 11016           & {11016} & 448        & 9226         \\ \hline
QAOA-Hard 4 & 33 & 2 & 47 & 60 & 31 & 3882 & 1000 & 33000 & 14019 & {14019} & 360 & 15102 \\ \hline
\multicolumn{1}{|l|}{QAOA-Par 1}                      & \multicolumn{1}{l|}{15}   & \multicolumn{1}{l|}{2}   & \multicolumn{1}{l|}{26}                    & \multicolumn{1}{l|}{46}       & \multicolumn{1}{l|}{20}     & \multicolumn{1}{l|}{943}      & \multicolumn{1}{l|}{201}        & \multicolumn{1}{l|}{3015}                 & \multicolumn{1}{l|}{2215}            & \multicolumn{1}{l|}{{2215}} & \multicolumn{1}{l|}{56}         & \multicolumn{1}{l|}{2109}         \\ \hline
\multicolumn{1}{|l|}{QAOA-Par 2}                      & \multicolumn{1}{l|}{15}   & \multicolumn{1}{l|}{2}   & \multicolumn{1}{l|}{26}                    & \multicolumn{1}{l|}{45}       & \multicolumn{1}{l|}{20}     & \multicolumn{1}{l|}{1009}     & \multicolumn{1}{l|}{201}        & \multicolumn{1}{l|}{3015}                 & \multicolumn{1}{l|}{2215}            & \multicolumn{1}{l|}{{2215}} & \multicolumn{1}{l|}{53}         & \multicolumn{1}{l|}{2114}         \\ \hline
\multicolumn{1}{|l|}{QAOA-Par 3}                      & \multicolumn{1}{l|}{18}   & \multicolumn{1}{l|}{2}   & \multicolumn{1}{l|}{29}                    & \multicolumn{1}{l|}{43}       & \multicolumn{1}{l|}{18}     & \multicolumn{1}{l|}{1080}     & \multicolumn{1}{l|}{201}        & \multicolumn{1}{l|}{3618}                 & \multicolumn{1}{l|}{2218}            & \multicolumn{1}{l|}{{2218}} & \multicolumn{1}{l|}{50}         & \multicolumn{1}{l|}{2023}         \\ \hline
QAOA-Par 4      & 15 & 2 & 29 & 29 & 25 & 3668 & 1000 & 15000 & 14001 & {14001} & 368 & 12897 \\ \hline
\end{tabular}
\captionof{table}{\label{tab:evaluation}Evaluation results. ASAP is the minimal depth of original loop body. $C$-ASAP is the minimal depth of the original loop body unrolled by $C$ times. Pre, K and Post represents prologue, kernel and epilogue. For each test case a range sized \#Iter is assigned, and the span of the output loop is QSP\#Iter.}
\end{center}
\end{figure*}

\subsection{Modulo scheduling again}
In the first round of modulo scheduling, inversion of $CZ$ and antidiagonal gates may introduce $Z$ gates overlapping $CZ$s, resulting an illegal schedule. To generate an executable schedule, we perform modulo scheduling again, but this time we no longer allow ``commutativity'' between antidiagonals and $CZ$s, and thus the inversion-fix step can be skipped. The scheduled loop by this second round of modulo scheduling is directly executable on the device.

[\textit{An analysis on the complexity of our algorithm presented in this section is given in Appendix \ref{complexity}.}]

\section{Evaluation}
\label{sec:evaluation}
We have implemented our method and carried out experiments on several quantum programs. Some of them are intrinsically parallel, while others are not.
Baselines for our evaluation come from the following sources:
\begin{itemize}
    \item \textbf{Kernel-ASAP} performs compaction and ASAP scheduling on the loop kernel. We expect our work to outperform this naive approach.
    \item \textbf{Unroll} unrolls the loop and performs compaction as well as ASAP scheduling on the unrolled circuit. The software-pipelined version should generate a program with similar depth but much smaller code size.
    \item \textbf{Cirq} uses the optimization passes in \cite{quantum_ai_team_and_collaborators_2020_4062499} to unroll the loop. This gives another perspective of loop unrolling besides our implementation.
\end{itemize}

The experiment results are in Table \ref{tab:evaluation}. We hereby analyze some of the important examples:

\subsection{Grover Search}
Grover search is a test case with long dependency chain and little space for optimization. Yet our approach can reduce the overall depth by merging adjacent gates in iteration and across iterations.
We use the $CCNOT$ case from \cite{DBLP:journals/corr/abs-1804-03719} and Sudoku solver from \cite{Qiskit-Textbook}. Since Grover search is a hard-to-optimize case, we inspected the optimized code and got the following findings: 

Although examples do not reveal much optimization chance, there is a pitfall for ASAP optimizers that may cause a diagonal $T^\dag$ gate to be scheduled at the first tick alone. This is prevented in our approach by performing bidirectional compactions. Moreover, the depth cut mainly comes from inversion of a pair of $CZ$s while scheduling, which indeed our approach does not consider. (see Figure \ref{fig:grover_inverse_cz}). This inspires us to find more optimization chances while placing instructions without dependency, like a program with many $CZ$s.
\begin{figure}
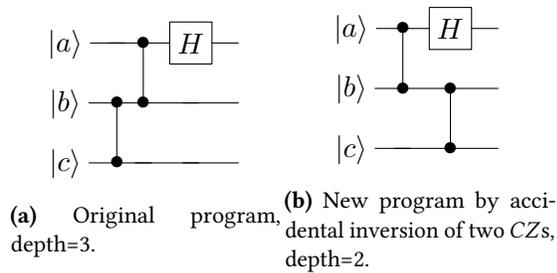

\begin{center}
    \begin{subfigure}{0.2\textwidth}
        \begin{center}
            \begin{qprognamed}{grover_inverse_cz_1}
            qubit a
            qubit b
            qubit c
            ZZ b,c
            ZZ a,b
            H a
            \end{qprognamed}
        \end{center}
        \caption{Original program, depth=3.}
    \end{subfigure}
    \begin{subfigure}{0.2\textwidth}
        \begin{center}
            \begin{qprognamed}{grover_inverse_cz_2}
            qubit a
            qubit b
            qubit c
            ZZ a,b
            ZZ b,c
            H a
            \end{qprognamed}
        \end{center}
        \caption{New program by accidental inversion of two $CZ$s, depth=2.}
    \end{subfigure}
    \caption{The accidental inversion of $CZ$s reduced kernel depth by 1.}
    \label{fig:grover_inverse_cz}
\end{center}
    
\end{figure}
\subsection{QAOA}
The QAOA programs in \cite{farhi2014quantum} (in Figure \ref{fig:qaoa_examples}), as well as the QAOA example in \cite{quantum_ai_team_and_collaborators_2020_4062499} are used in our experiment, but with a $p$ (i.e. the number of iterations) large enough.
Since the decomposition of QAOA into gates affects how it can be optimized on our architecture, we consider two different ways: \textbf{QAOA-Par} where QAOA is decomposed to expose more commutativity (\textit{see the details in Appendix \ref{appendix:par_qaoa}}), and \textbf{QAOA-Hard}, where QAOA is decomposed into a harder form, with a long dependency chain formed by cross-qubit operations that is unable to be detected by gate-level optimizers.

The evaluation results in Table \ref{tab:evaluation} show that in all cases, our approach can reduce the loop kernel size compared with \textbf{Kernel-ASAP}, and can sometimes outperform unrolling results. This advantage is more evident in the QAOA-Par cases than in the QAOA-Hard cases, since QAOA-Par reveals more commutativity chances than QAOA-Hard. Another finding is that QAOA-Hard generates larger code than QAOA-Par, and thus requires more iterations for software-pipelining to take effect.

\begin{figure}
\begin{center}
\begin{subfigure}{0.08\textwidth}
        \centering\includegraphics[max width=\textwidth]{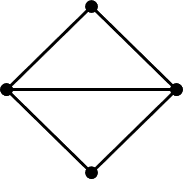}
    \caption{}
\end{subfigure}
\hspace{0.5cm}
\begin{subfigure}{0.12\textwidth}
        \centering\includegraphics[max width=\textwidth]{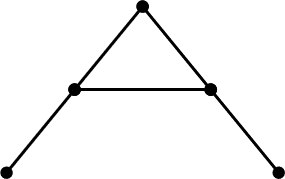}
    \caption{}
\end{subfigure}
\hspace{0.5cm}
\begin{subfigure}{0.12\textwidth}
        \centering\includegraphics[max width=\textwidth]{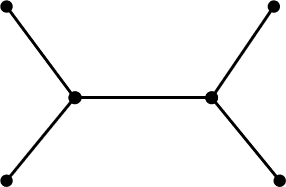}
    \caption{}
\end{subfigure}
\end{center}
\caption{\label{fig:qaoa_examples}QAOA-MaxCut examples in \cite{farhi2014quantum}.}
\end{figure}
{\color{red}
}
\textit{[More discussions on examples are in Appendix \ref{cluster-optimization}.]}
\section{Conclusion}
\label{sec:conclusion}
We proposed a compilation flow for optimizing quantum programs with control flow of for-loops. In particular, data dependencies and resource dependencies are redefined to exposes more chances for optimization algorithms. 
Our approach is tested against several important quantum algorithms, revealing code-size advantages over the existing approaches while keeping depth advantage close to loop rolling. Yet there is still gap for optimization of more complex quantum programs, on different architectures, and with lower complexity, which could be filled in future works.



\newpage

\bibliographystyle{ACM-Reference-Format}
\bibliography{references}


\newpage

\appendix

\section{Basic quantum gates}\label{appendix:basic_quantum_gates}

The following are the frequently-used one-qubit gates represented in $2\times 2$ unitary matrices:
    \begin{align*}{\rm Pauli\ gates}:\ 
    &X=\mattwo{0}{1}{1}{0},\\ &Y=\mattwo{0}{-i}{i}{0},\\ &Z=\mattwo{1}{0}{0}{-1},\\
  {\rm Hadamard\ gate}:\ &H=\frac{1}{\sqrt{2}}\mattwo{1}{1}{1}{-1},\\ 
{\rm Phase\ and}\ \frac{\pi}{8}\ {\rm gates}:\    
   &S=\mattwo{1}{0}{0}{i},\\ &T=\mattwo{1}{0}{0}{e^{\frac{i\pi}{4}}},\\
   {\rm Pauli\ Rotations}:\ &R_X(\alpha)=\mattwo{cos\frac{\alpha}{2}}{-isin\frac{\alpha}{2}}{-isin\frac{\alpha}{2}}{cos\frac{\alpha}{2}},\\
    & R_Y(\alpha)=\mattwo{cos\frac{\alpha}{2}}{-sin\frac{\alpha}{2}}{sin\frac{\alpha}{2}}{cos\frac{\alpha}{2}},\\ 
    &R_Z(\alpha)=\mattwo{e^{-\frac{i\alpha}{2}}}{0}{0}{e^{\frac{i\alpha}{2}}}.\end{align*} They combined with one of the (two-qubit) controlled gates 
    \begin{align*}
    \mathit{CNOT}&=\left[ \begin{array}{cccc}1& & & \\ &1& &  \\ & & &1 \\ & &1& \end{array}\right],\\ CZ&=\left[ \begin{array}{cccc}1& & & \\ &1& &  \\ & & 1& \\ & & &-1 \end{array}\right].
\end{align*} are universal for quantum computing; that is, they can be used to construct arbitrary quantum gate of any size.

Beside the above, we will use the following auxiliary gates to simplify the presentation of our approach:
\begin{equation*}\begin{split}R_X^{-}(\alpha)&=\left[\begin{array}{cc}\cos\frac{\theta}{2}& -i\sin\frac{\theta}{2}\\
i\sin\frac{\theta}{2}&-\cos\frac{\theta}{2}\end{array}\right],\\
R_Z^{+}(\alpha)&=\left[\begin{array}{cc}0 &
e^{i\alpha/2}\\e^{-i\alpha/2}
&0\end{array}\right]=XR_Z(\alpha),\\
H(\alpha)&=\frac{1}{\sqrt{2}}\left[\begin{array}{cc}1& 1\\
e^{i\alpha}& -e^{i\alpha}\end{array}\right]=R_Z(\alpha)H,\\
H^{-}(\alpha)&=\frac{1}{\sqrt{2}}\left[\begin{array}{cc}1& -1\\
e^{i\alpha}&e^{i\alpha}\end{array}\right]=R_Z(\alpha)HZ.
\end{split}\end{equation*}
Note that parameter $\alpha$ in the above gates is a real number. The $R_Z^{+}(\alpha)$ gate can represent all single qubit gates that are anti-diagonal, i.e. only anti-diagonal entries are not $0$. The other three notations are used in Appendix \ref{appendix:cnot_conjugation}.

For real-world quantum computers, a quantum device may only support a discrete or contiguous set of single qubit gates while keeping the device universal.
For example, IBM's devices allow the following three kinds of single qubit gates to be executed directly\cite{Qiskit-Textbook}:
\begin{equation*}
\begin{aligned}
&U_1(\lambda)=\mattwo{1}{0}{0}{e^{i\lambda}},\\ &U_2(\phi, \lambda)=\frac{1}{\sqrt{2}}\mattwo{1}{-e^{i\lambda}}{e^{i\phi}}{e^{i\lambda+i\phi}},\\
&U_3(\theta, \phi, \lambda)=\mattwo{cos(\frac{\theta}{2})}{-e^{i\lambda}sin(\frac{\theta}{2})}{e^{i\phi}sin(\frac{\theta}{2})}{e^{i\lambda+i\phi}cos(\frac{\theta}{2})}
\end{aligned}
\end{equation*}
Note that $U_2(\phi, \lambda)=U_3(\frac{\pi}{2}, \phi, \lambda)$ and $U_1(\lambda)=U_3(0,0,\lambda)$. Also note that gate $U_3$ itself is universal for single-qubit gates, and the main reasons for supporting $U_1$ and $U_2$ is to mitigate error, which is beyond our consideration.
\section{More Examples for quantum loop programs}\label{more-examples}

We hereby presents more quantum algorithms that can be written into quantum loop programs and can thus be potentially optimized by our approach.

\subsection{One-way quantum computing}

Preparation circuit for simulating one-way quantum computation on quantum circuit is another example that allows each iteration to be performed on different qubits.

\begin{example}
One-way quantum computing $QC_\mathcal{C}$\cite{raussendorf2003measurement} is a quantum computing scheme that is quite different from the commonly used quantum-circuit based schemes. Instead of starting from $\ket{0}$, $QC_\mathcal{C}$ initializes all qubits (on a 2-dimensional qubit grid) in a highly-entangled state, called \textbf{cluster state}. After the preparation step, $QC_\mathcal{C}$ performs single-qubit measurements on all qubits and extract the computation result from these measurement outcomes.

To simulate one-way quantum computing with quantum circuit, we first need to prepare the cluster state from $\ket{0}$. This can be done by first performing Hadamard gates on all qubits, then performing $CZ$ gate on each pair of adjacent qubits on the qubit grid.

The preparation circuit can be written in a nested loop manner. If we assume the grid has a fixed width ($3$ in our case), we can unroll the innermost loop to get the flattened loop:

\begin{algorithmic}
\STATE {$H[q[0]]$}
\STATE {$H[q[1]]$}
\STATE {$H[q[2]]$}
\STATE {$CZ[q[0],q[1]]$}
\STATE {$CZ[q[1],q[2]]$}
\FOR{i=1 \TO (L-1) }  {
    \STATE {$H[q[3i]]$}
    \STATE {$H[q[3i+1]]$}
    \STATE {$H[q[3i+2]]$}
    \STATE {$CZ[q[3i],q[3i+1]]$}
    \STATE {$CZ[q[3i+1],q[3i+2]]$}
    \STATE {$CZ[q[3i],q_1[3i-3]]$}
    \STATE {$CZ[q[3i+1],q_2[3i-2]]$}
    \STATE {$CZ[q[3i+2],q_3[3i-1]]$}
} \ENDFOR
\end{algorithmic}

Figure \ref{fig:cluster-state-preparation} shows the gates and qubits involved in each iteration where $L=5$.
The optimization of this program will be discussed in Appendix \ref{cluster-optimization}. 
\end{example}
\begin{figure*}
    \begin{center}\includegraphics[max width=0.8\textwidth]{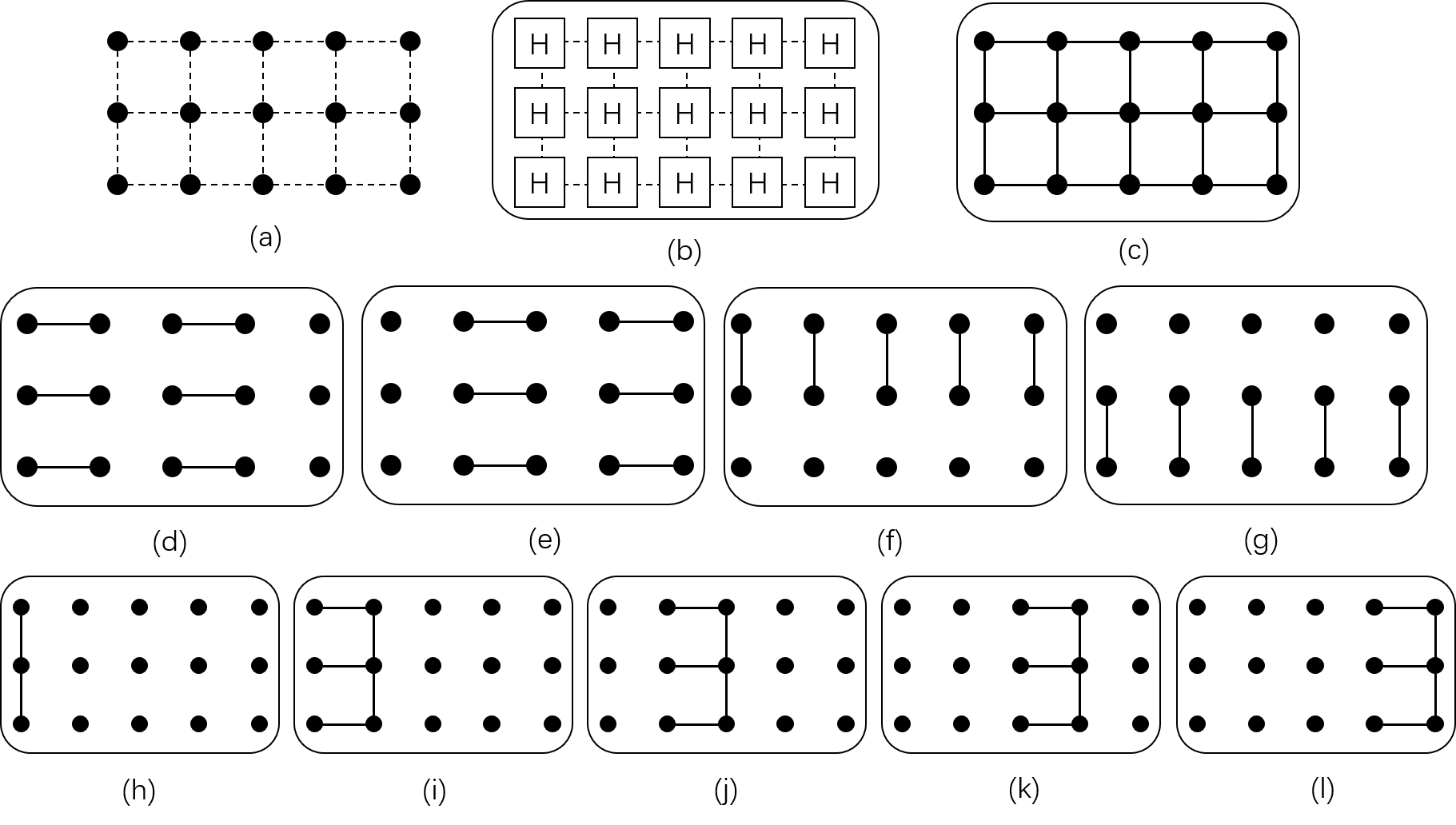}
    \end{center}
    \caption{\label{fig:cluster-state-preparation}Converting cluster state preparation circuit into loop program.
    Fig (a) is a $3\times 5$ two-dimensional qubit network. The preparation is done by performing a layer of Hadamard gates (Fig (b)) and a layer of $CZ$ gates (Fig (c)). 
    One way to perform those $CZ$ gates without qubit conflict is to split them into four non-overlapping groups and execute each group separately, as in Fig (d) to Fig (g).
    The procedure can also be written into loop program, as in Fig (h) to Fig (l).
    }
\end{figure*}

\subsection{Quantum Approximate Optimization Algorithm}
\begin{example}
Quantum Approximate Optimization Algorithm (QAOA)\cite{farhi2014quantum} can be used to solve MaxSat problems, for example, MaxCut problems on 3-regular graphs, say $G=\langle V,E\rangle$. QAOA performs quantum computation and classical computation alternatively. On the quantum part, it requires us to create the state: 
\begin{equation}
    \ket{\gamma, \beta}=\prod_{i=1}^{p}U(B, \beta_i)U(C, \gamma_i)\ket{+}
\end{equation}
where: 
\begin{align}
    &U(C, \beta_i)=\prod_{(a,b)\in E}\left[ \begin{array}{cccc}1& & & \\ &e^{-i\omega_{ab}\gamma_i}& &  \\ & & e^{-i\omega_{ab}\gamma_i} & \\ & & &1 \end{array}\right]\\
    &U(B, \gamma_i)=\prod_{j=0}^{n-1}R_X(\beta_i,j).
\end{align}

The sets of parameters $\left\{\beta_i\right\}$ and $\left\{\gamma_i\right\}$ are computed in the classical computation between every two quantum epochs. This requires the optimizer to support compilation of the circuit above without knowing all parameters in advance.

$U(B, \gamma_i)$ are products of Pauli $X$ rotations on all qubits. Since in our case $U(C, \beta_i)$ can be decomposed in the following way:
\begin{equation}
    \label{equ:qaoa_impl}
    \left[ \begin{array}{cccc}1& & & \\ &e^{-i\omega_{ab}\gamma_i}& &  \\ & & e^{-i\omega_{ab}\gamma_i} & \\ & & &1 \end{array}\right]=
    \begin{qprognamed}{qaoa_implementation_of_gamma_rot}
        qubit a
        qubit b
        def	ua,0,'R_Z(-\omega_{ab}\gamma_i)'
        cnot a,b
       ua b
        cnot a,b
   \end{qprognamed}
    ,
\end{equation}
we can define parametric gate arrays $U_C[i]=R_X(\beta_i,j)$ and $U_B[i]=R_Z(-\omega_{ab}\gamma_i)$, and the QAOA quantum part can be written as a parametric quantum loop program:
\begin{algorithmic}
\FOR{i=0 \TO (N-1) }  {
    \STATE {$H[q[i]]$}
} \ENDFOR
\FOR{i=1 \TO p }  {
    \FOR{$(a, b) \in E$}  {
        \STATE {$CNOT[q[a], q[b]]$}
        \STATE {$U_B[i][q[b]]$}
        \STATE {$CNOT[q[a], q[b]]$}
    }\ENDFOR
    \FOR{j=0 \TO (N-1) }  {
        \STATE {$U_C[j][q[j]]$}
    } \ENDFOR
} \ENDFOR
\end{algorithmic}
\end{example}

The two nested loops can be fully unrolled by hand, and the outcome loop satisfies our requirements for optimization.

\section{Output language}
\label{appendix:output_language}
If the input range of the loop program is unknown, we may have to add guard statements into the orginal program, for example, when we want to check if the range is large enough for us to use the software-pipelined version. Those features such as guard statements, unfortunately, are not supported in our definition of input language. So we have to define the following language for the optimization result:
\begingroup
\allowdisplaybreaks
\begin{align*}
    \textbf{program} :=& \textbf{header}\ \textbf{statement}*\\
    \textbf{header} :=& [(\textbf{qdef}\ |\ \textbf{udef})^*]\\
    \textbf{qdef} :=& qubit\ \textbf{ident}[\mathbb{N}];\\
    \textbf{udef} :=& defgate\ \textbf{ident}[\mathbb{N}] = \textbf{gate};\\
    \textbf{gate} :=& [(\mathbb{C}^{2\times 2})^*]\ |\ R_Z\ |\ R_Z^+\ |\ Unknown\\
    \textbf{gateref} :=& \textbf{ident}[\textbf{expr}]\\
    \textbf{qubit} :=& \textbf{ident}[\textbf{expr}]\\
    \textbf{op} :=& SQ(\textbf{gateref})\ \textbf{qubit};\\
                         |\ & CZ\ \textbf{qubit},\textbf{qubit};\\
    \textbf{statement} :=& \textbf{op}\\
                         |\ & for\ \textbf{ident}\ in\ \textbf{expr}\ to\ \textbf{expr} \{\textbf{statement}^*\}\\
                         |\ & parallel\{\textbf{statement}*\}\\
                         |\ & guard\{\\
                             &(\textbf{compare}=>\{\textbf{statement}^*\})^* \\ %
                         &otherwise=>\{\textbf{statement}^*\}\\
                             &\} \\
    \textbf{expr} :=&\textbf{ident}\ |\ expr+expr\ |\ expr-expr\\
                 |\ &expr*expr\ |\ expr/expr\ |\ expr\%expr\ |\ \mathbb{Z} \\
    \textbf{compare} :=& \textbf{expr}\ \textbf{ordering}\ \textbf{expr}\\
    \textbf{ordering} :=& ==\ |\ !=\ |\ >\ |\ <\ |\ >=\ |\ <=
\end{align*}
\endgroup
The main differences between the input language and the output language are:
\begin{enumerate}
    \item The $parallel$ notation is added to explicitly point out which instructions are scheduled together.
    \item The $guard$ statement is added to check whether the input range is suitable for the software-pipelined version if the range is unknown at compilation time, and to separate cases with different $(m\ mod\ C)$. The $guard$ statement executes the first statement block with a satisfied guard condition.
    \item The $expr$ allows for more general indexing into qubit arrays and gate arrays. Note that the division and modulo operators are Euclidean, i.e. it always holds that 
    \begin{equation}
        \left\{\begin{aligned}
        &sign(a\%b)=sign(b)\\
        &a\%b + (a/b)*b = a
        \end{aligned}\right.
    \end{equation}
\end{enumerate}
\section{Solving Diophantine equations}
\label{appendix:diophantine-noilp}
In this appendix we focus on solving the Diophantine equation:
\begin{equation}
    (k_2-k_1)i+k_2(\Delta i)=b_1-b_2, i\in T, i+\Delta i \in T, \Delta i\geqslant 1.
\end{equation} 
We rewrite it into:
\begin{equation}
    ax+by=c, x\in T, x+y\in T, y\geqslant 1.
\end{equation}

We recall the solutions $S$ for linear Diophantine equations with two variables:
\begin{lemma}
\textbf{Solutions for linear Diophantine equations with two variables}
\begin{equation}
    ax+by=c, x\in\mathbb{Z}, y\in \mathbb{Z}.
\end{equation}
\begin{enumerate}
    \item If $a=0$ and $b=0$, $S=\Phi$ if $c\neq 0$ and $S=\mathbb{Z}\times \mathbb{Z}$ if $c=0$.
    \item If $a=0$ but $b\neq 0$ (similar for $b=0$ but $a\neq 0$),
    \begin{enumerate}
        \item If $b | c$, $S=\mathbb{Z}\times \left\{\frac{c}{b}\right\}$.
        \item Otherwise, $S=\Phi$.
    \end{enumerate}
    \item If $a\neq 0$ and $b\neq 0$:
    \begin{enumerate}
        \item If $c=d\cdot gcd(a,b)$,
        \begin{itemize}
            \item Special solution $(x_0,y_0)$ where \begin{equation}ax_0+by_0=gcd(a,b)\end{equation} can be solved using extended Euclidean algorithm.
            \item General solution $\left(k\frac{b}{gcd(a,b)}, -k\frac{a}{gcd(a,b)}\right)$ for equation \begin{equation}ax+by=0\end{equation} is known.
            \item The total solution space is
            \begin{equation}
                S=\left\{\left(x_0+k\frac{b}{gcd(a,b)}, y_0-k\frac{a}{gcd(a,b)}\right) | k\in \mathbb{Z}\right\}.
            \end{equation}
            We rewrite the equation into:
            \begin{equation}
                S=\left\{\left(x_0+k\Delta x, y_0+k\Delta y\right) | k\in \mathbb{Z}\right\}.
            \end{equation}
        \end{itemize}
        \item Otherwise, $S=\Phi$.
    \end{enumerate}
\end{enumerate}
\end{lemma}
For our original question with constraints, we only consider the cases where $a\neq 0$ and $b\neq 0$.

When $T=\mathbb{Z}$, the constraints no longer exist and we only need to find the minimal positive integer in set $\left\{y_0+k\Delta y\right\}$, which can be solved by an Euclidean division. With loss of generality, we can just let $k=0$ by choosing $y_0$ to be exactly the smallest positive integer in $\left\{y_0+k\Delta y\right\}$ and adjust $x_0$ accordingly, without affecting the solution set $S$.

When $T=\left[a, b\right]$, the corresponding $x_0$ may not lie in $T$. In this case we may want to find a secondary-minimal positive integer. Without loss of generality we assume $\Delta y>0$ (otherwise choose $\Delta x = - \Delta x$ and $\Delta y = - \Delta y$). Then the problem becomes: find minimal $k\in N_+$ s.t.
\begin{equation}
    \left\{\begin{aligned}
        &x_0+k\Delta x >= a\\
        &x_0+k\Delta x <= b
    \end{aligned}\right.,
\end{equation}
which is equivalent to
\begin{equation}
    \left\{\begin{aligned}
        &k\Delta x >= a-x_0\\
        &k\Delta x <= b-x_0
    \end{aligned}\right.
\end{equation}
which can thus be solved by a routine calculation: a minimal $k$ exists, or does not exist at all.

\section{Proofs of Theorems \ref{rule-1} (CZ conjugation rules)}
In this section we give out proof for our new rules of instruction data dependency. We will show that our definition of dependency is ``sufficient and necessary'' for quantum gate sets using $CZ$.

\label{appendix:cz_conjugation_rules}
We first restate Theorem \ref{rule-1} as follows:
   $$CZU_AU_BCZ=V_AV_B,$$ if and only if $U_A$ and $U_B$ are diagonal or anti-diagonal. That is, $U_i=R_Z(\theta)$ or $U_i=R_Z^+(\theta)$ for $i\in\{A,B\}$.
    \begin{proof}
    We here introduce our methodology of proving quantum gate algebra equations: first we give a necessary condition by trying several input states, and show that the condition is also sufficient for the equation to hold.
    
    The first lemma is a criteria for deciding whether a state is separable or entangled:

    \begin{lemma}
        \label{lem:entangled-state-decide}
        Two-qubit state
        $\ket{\psi}=\left(a, b, c, d\right)^T$ is separable if and only if:
        \begin{equation}
            ad-bc=0.
        \end{equation}
    \end{lemma}
        \begin{proof}
        \textbf{(Necessity)} If $\ket{\psi}$ is separable,
        there exists two single qubit states $\ket{\psi_1}$ and $\ket{\psi_2}$,
        s.t. 
        \begin{equation}
            \ket{\psi}=\ket{\psi_1}\otimes \ket{\psi_2}
        \end{equation}
        Suppose
        \begin{equation}
            \ket{\psi_1}=\left(\alpha_1, \beta_1\right)^T,
        \end{equation}
        \begin{equation}
        \ket{\psi_2}=\left(\alpha_2, \beta_2\right)^T,
        \end{equation}
        We have
        \begin{equation}
            \ket{\psi}=\left(\alpha_1\alpha_2, \alpha_1\beta_2, \beta_1\alpha_2, \beta_1\beta_2\right)^T,
        \end{equation}
        and it can be easily verified that $ad-bc=0$.    

        \textbf{(Sufficiency)} If
        \begin{equation}\ket{\psi}=\left(a,b,c,d\right)^T\end{equation} with $ad-bc=0$,
        \begin{enumerate}
        \item If $b=0$, this indicates $a=0$ or $d=0$. If $a=0$, let
        \begin{equation}
        \begin{cases}
        \ket{\psi_1}=\ket{1}\\
        \ket{\psi_2}=c\ket{0}+d\ket{1}
        \end{cases};
        \end{equation}
        otherwise $d=0$, and let
        \begin{equation}
        \begin{cases}
        \ket{\psi_1}=a\ket{0}+c\ket{1}\\
        \ket{\psi_2}=\ket{0}
        \end{cases}.
        \end{equation}
        \item If $c=0$, this indicates $a=0$ or $d=0$. If $a=0$, let
        \begin{equation}
        \begin{cases}
        \ket{\psi_1}=b\ket{0}+d\ket{1}\\
        \ket{\psi_2}=\ket{1}
        \end{cases};
        \end{equation}
        otherwise $d=0$, and let
          \begin{equation}
        \begin{cases}
        \ket{\psi_1}=\ket{0}\\
        \ket{\psi_2}=a\ket{0}+b\ket{1}
        \end{cases}.
        \end{equation}
        \item Otherwise $a,b,c,d\neq 0$. Let
        \begin{equation}
        \begin{cases}
        \ket{\psi_1}=\left(\frac{b}{\sqrt{\|b\|^2+\|d\|^2}},\frac{d}{\sqrt{\|b\|^2+\|d\|^2}}\right)^T\\
        \ket{\psi_2}=\left(\frac{a}{b\sqrt{\|\frac{a}{b}\|^2+\|1\|^2}},\frac{1}{\sqrt{\|\frac{a}{b}\|^2+\|1\|^2}}\right)^T
        \end{cases}.
        \end{equation}
        It can be verified that $\|\ket{\psi_1}\|=\|\ket{\psi_2}\|=1$, and that
        \begin{equation}
            \ket{\psi_1}\otimes \ket{\psi_2} = \frac{\left(a,b,c,d\right)^T}{\sqrt{(\|b\|^2+\|d\|^2)(\|\frac{a}{b}\|^2+\|1\|^2})},
        \end{equation}
        which is exactly $(a,b,c,d)^T$ since tensor product preserves norm.
        \end{enumerate}
        \end{proof}

    \begin{lemma}
    \textbf{(Necessity)} For the equation to hold, $U_A$ and $U_B$ have to be diagonal or anti-diagonal. This means $U_i$ transforms $\ket{0}$ to $\ket{0}$ or $\ket{1}$, up to a global phase.
\end{lemma}
    \begin{proof}
    Suppose $\ket{\phi}=U_A\ket{0}=(a, b)^T$, thus
    \begin{align}
        &CZU_AU_BCZ(\ket{0}\otimes(U_B^\dag\ket{\phi}))\\
        =&CZ\ket{\phi}\otimes\ket{\phi}\\
        =&(a^2, ab, ab, -b^2)^T,
    \end{align} which should be a separable state since this is also $V_AV_B(\ket{0}\otimes(U_B^\dag\ket{\phi}))$, which is separable. Thus $a^2b^2=0$, so $a=0$ ($R_Z^+$ case) or $b=0$ ($R_Z$ case). This is the same for $U_B$.
    \end{proof}
    
    \begin{lemma}
    \textbf{(Sufficiency)} $R_Z$ and $R_Z^+$ satisfies the conjugation rules.
\end{lemma}
    \begin{proof}
    Note that $R_Z^+=XR_Z$ and $CZX_A=X_AZ_BCZ$. By simple computation we can see the conjugation holds.
    \end{proof}
    
    \end{proof}

\section{Proof of Theorem \ref{thm:compaction} (Convergence of compaction)}
\label{appendix:compaction}
We show that compaction procedure will converge after applying the procedure three times.

If we look at the factors that prevents compaction procedure from reaching its fixpoint, there are two main reasons:
\begin{enumerate}
    \item Single qubit merging results in new diagonal gates or antidiagonal gates, which is not recognized when the first gate is placed. Compacting $\#1$ in Figure \ref{fig:greedy-leftnormalize-bad} shows an example where three gates merge into an antidiagonal $X$ gate, which can merge through the $CZ$ gate on next compaction.
    \item Antidiagonal and $CZ$ changing order will add $Z$ gates to the circuit. Compacting $\#2$ in Figure \ref{fig:greedy-leftnormalize-bad} shows an example.
\end{enumerate}

Fortunately, these problems will not occur at the third time of compaction. This is because diagonal gates and antidiagonal gates forms a subgroup of $U_2$:
\begin{lemma}
    \label{lem:zgroup}
    Let
    \begin{align}
    G_{Z}&=\left\{R_Z(\theta)|\theta \in \left[0, 2\pi\right)\right\},\\
    G_Z^+&=\left\{R_Z^+(\theta)|\theta \in \left[0, 2\pi\right)\right\},\\
    G&=G_Z\cup G_Z^+,
    \end{align}
    thus $G_Z$,$G$ are subgroups of $U_2$, while $\forall g_1,g_2\in G_Z^+, g_1g_2\in G_Z$.
    \end{lemma}

    \begin{corollary}
        \label{cor:no-new-diagonal}
    $\forall g_1\in U_2\backslash G, g_2 \in G, g_2g_1\in U_2\backslash G$.
    \end{corollary}

On $\#2$ compaction, single qubit gates can only merge when they are on different sides of a $CZ$ gate and one is diagonal or antidiagonal (otherwise they should have been merged on $\#1$ compaction). According to corollary \ref{cor:no-new-diagonal}, this merging will not add new diagonals or antidiagonals, and all new gates from compaction $\#2$ come from moving antidiagonal through $CZ$. The last compaction merges these additional $Z$ gates to their left.
\section{Proof of Theorem \ref{th:remove-multiple-edge} (Remove multiple edges)}
\label{appendix:remove-multiple-edges}
In the QDG defined in Section \ref{sec:resched_loop_body}, Theorem \ref{th:remove-multiple-edge} is proposed so that multiple edges can be removed before $II$ is assigned. The proof of Theorem \ref{th:remove-multiple-edge} is listed below:

\begin{proof}
    Since $dif_1$ and $dif_2$ are integers,
    \begin{equation}
        1+dif_2\leqslant dif_1,
    \end{equation}
    Since $II\geqslant 1$,
    \begin{equation}
        \label{equ:remove_multiedge_e1}
        -II\cdot dif_1\leqslant -II-II\cdot dif_2\leqslant -1-II\cdot dif_2.
    \end{equation}
    Since $min_1\leqslant 1$ and $min_2\leqslant 1$,
    \begin{equation}
        \label{equ:remove_multiedge_e2}
        min_1\leqslant min_2+1.
    \end{equation}
    Adding up Equation \ref{equ:remove_multiedge_e1} and \ref{equ:remove_multiedge_e2} shows the result.
\end{proof}
\section{Resource scheduling complexity analysis}
\label{appendix:resource-scheduling-complexity}
In Secion IV we mentioned that we can keep retrying if there is a ``resource conflict'' and the death countdown is not timed-out (i.e. resource conflict are all caused by false conflicts), which may lead to too many retries that may dominate the complexity of the algorithm. This requires us to give an upper bound of maximum number of retries to estimate the total complexity.

Recall how we perform resource checking when inserting instructions into the schedule:

\begin{itemize}
    \item For every time slot, we have scheduled a bunch of instructions in this time slot.
    \item When adding an instruction or a group of instructions, we check the operands of each instruction to be added against instructions in the time slot where it will be added.
    \item If there is a resource conflict, we have to try next tick (and perhaps start a death countdown).
\end{itemize}

We first show that if there is only false conflict, the loop can be written into an equivalent form where all $k=1$. In fact, this is achieved by the fact:
\begin{equation}
    ki+b=k(i+(b/k))+(b\mod k),
\end{equation}
where
\begin{equation}
    (b \mod k)\in\left[0, \|k\|\right), k(b / k)+(b\mod k)=b.
\end{equation}

According to this fact, the array can be split into $\|k\|$ slices, and resource conflict can occur if the two qubit references fall into the same slice. Figure \ref{fig:qubit-array-splitting} is an example for $k=3$. Offsets $3i$ and $(3i-1)$ will never conflict with each other, since they fall into different slices $q_0$ and $q_2$.

This splitting allows us to use one integer $b'=(b/ k)$ to represent an expression in the slice: in the Figure \ref{fig:qubit-array-splitting} case we can use $0$ for $q[3i]$ in slice $q_0$, $0$ for $q[3i+1]$ in slice $q_1$, and $(-1)$ for $q[3i-1]$ in slice $q_2$.

\begin{corollary}
For the modulo scheduling, if a resource is scheduled $II$ ticks later, the integer $b'$ representing the resource decreases by 1.
\end{corollary}

\begin{figure*}
    \begin{center}
    \includegraphics[width=\textwidth]{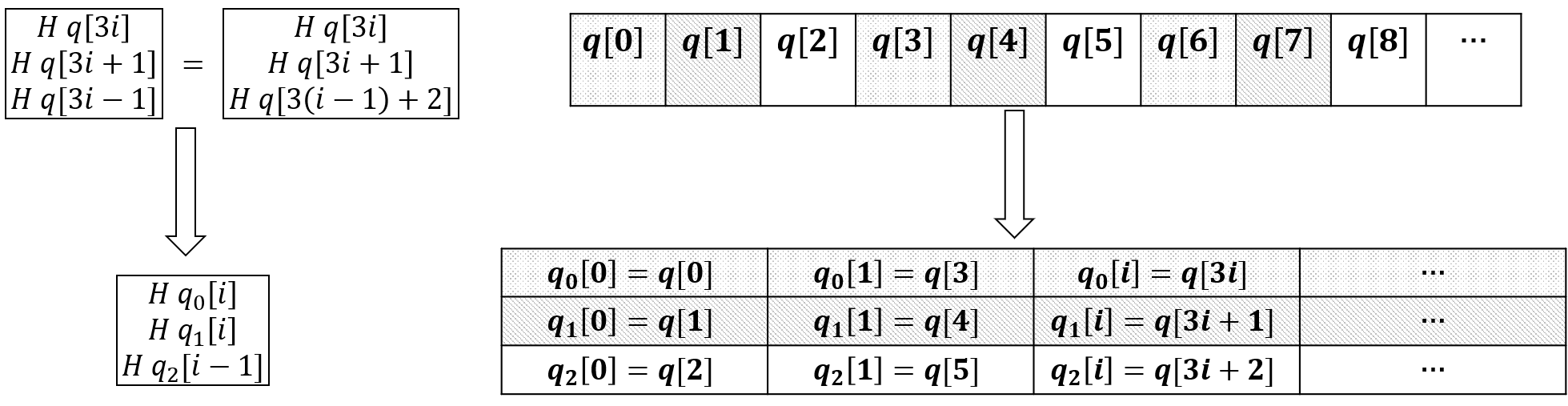} 
    \end{center}
    \caption{Example for splitting the qubit array when $k=3$. Resource conflict can only occur inside each slice, and resources in each slice can be represented by one integer.}
    \Description{Resource conflict can only occur inside each slice.}
    \label{fig:qubit-array-splitting}
\end{figure*}

This allows to use a stricter model for upper-bound estimation:
\begin{itemize}
    \item For the entire schedule, we use a universal set to store all integer representations $\left\{b'\right\}$ of linear expressions.
    \item When adding an instruction or a group of instructions, we check the operands to be added against the universal set, rather than the time-slot set. This means two instructions with the same operand but scheduled at different ticks will also be seen as conflicted.
    \item If the integer representation of operand is already in the set, there is a resource conflict. To find the worst case, we suppose the next $(II-1)$ tries will definitely fail. The next retry that will possibly success is the $II$-th retry where the instruction is going to be placed in the same time slot again.
    \item The array index $q$ and slice index $b\mod k$ are ignored. For example, operands $q[3i]$ and $q[3i+1]$ will be seen as conflicted since they have the same representation $0$, even though the two expressions will never be equal to each other.
\end{itemize}

This strict set of rules reduces our upper bound problem to a clearer problem:

\begin{theorem}
For finite set $A\subset Z$ standing for resources (integers representing each resource) already scheduled, and $B\subset Z$ being resources to be scheduled. Define
\begin{equation}
    B-(k\in N)=\left\{x-k | x\in B \right\}
\end{equation}
to be the resource set of $B$ after $kII$ retries.
Let $k_{min}$ be the minimal $k$, s.t.
\begin{equation}
    A\cap (B-k) = \Phi,
\end{equation}
then $k_{min}II$ retries is required at most in our algorithm.
\end{theorem}

\begin{figure}[H] 
    \centering
    \includegraphics[width=0.5\textwidth]{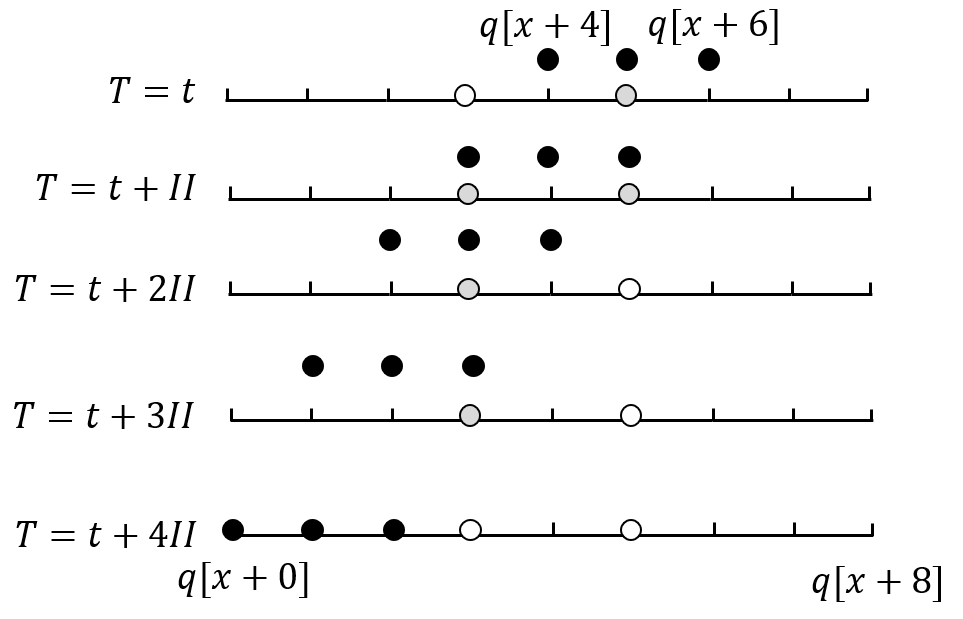}
	\caption[Scheduling conflict and retrying]{Resource 3($q[x+3]$) and 5($q[x+5]$) are now occupied, and resource 4 to 6 required to scheduled. Now $k_{min}=4$.}
    \label{fig:res-sched-axis}
\end{figure}

A naive estimation of $k_{min}$ would be 
\begin{equation}
    k_{min}\leqslant max(B)-min(A),
\end{equation}
which is not acceptable. Fortunately, we can give out a more precise estimation not in the values in $A$ or $B$, but only in the size of sets.

\begin{theorem}
Let $\|S\|$ be size of set $S$,
\begin{equation}
    k_{min}\leqslant \|A\|\|B\|.
\end{equation}
\end{theorem}
\begin{proof}
    Consider the set
\begin{equation}
    D=\left\{b-a | a\in A, b\in B, (b-a)\geqslant 0\right\}.
\end{equation}
thus $k\not\in D$ if and only if $A\cap(B-k)=\Phi$.
Thus $k_{min}$ is the first natural number not appearing in $D$. However, $\|D\|\leqslant\|A\|\|B\|$ according to its definition, so $k\leqslant \|A\|\|B\|$.
\end{proof}

\begin{corollary}
    Inserting $m$ instructions at one time (e.g. merging to scheduled blocks) into a schedule with $n$ instructions requires at most $O(mnII)$ retries. If each retry takes $O(mn)$ queries to find a conflict, the total complexity is at most $O(m^2n^2II)$.
\end{corollary}

According to the theorem, we can get some several important results on the complexity:

\begin{corollary}
\begin{enumerate}
    \item Inserting one instruction into the modulo scheduling table sized $b$ requires $O(bII)$ retries and $O(b^2II)$ time. Thus inserting all $b$ instructions require $O(b^3II)$ time. 
    \item The span of the modulo scheduling table above is bounded by $O(b^2II)$.
    \item Suppose the loop kernel sized $n$ is split into $a\geqslant 2$ strong connected components sized $b$, the total complexity for scheduling all SCCs is $aO(b^3II)=O(ab^3II)=O(n^4)$, and the total time required to merge all SCCs together is
    \begin{equation}
        \sum_{i=1}^{a-1}O(b^2(ib)^2II)=O(a^3b^4II)=O(n^5).
    \end{equation}
    \item  The span of the total schedule is
    \begin{equation}
        aO(b^2II)+\sum_{i=1}^{a-1}b(ib)II=O(ab^2II+a^2b^2II)=O(n^2II).
    \end{equation}
    Thus we expect the length of prologue and epilogue to be
    \begin{equation}
        \sum_{i=1}^{O(n^2)}i\cdot II = O(n^3).
    \end{equation}
\end{enumerate}
\end{corollary}

\section{CNOT conjugation rules}
\label{appendix:cnot_conjugation}
These results are taken directly from \cite{ying2007algebra}.

\begin{theorem}
    ($CNOT$ conjugation) $CNOT$ conjugates single qubit gates if and only if the conjugation satisfies one of the following eight cases:
    \begin{enumerate}
        \item 
    \begin{equation}
        \begin{qprognamed}{cnot_conjugation_1}
        qubit a
        qubit b
        def	ua,0,'R_Z(\alpha)'
        cnot a,b
        ua a
        cnot a,b
        \end{qprognamed}
        =
        \begin{qprognamed}{cnot_conjugation_2}
        qubit a
        qubit b
        def	uc,0,'R_Z(\alpha)'
        uc a
        \end{qprognamed}
    \end{equation}
        \item 
    \begin{equation}
        \begin{qprognamed}{cnot_conjugation_3}
        qubit a
        qubit b
        def	ua,0,'R_Z^+(\alpha)'
        cnot a,b
        ua a
        cnot a,b
        \end{qprognamed}
        =
        \begin{qprognamed}{cnot_conjugation_4}
        qubit a
        qubit b
        def	uc,0,'R_Z^+(\alpha)'
        def ud,0,'X'
        uc a
        ud b
        \end{qprognamed}
    \end{equation}
        \item 
    \begin{equation}
        \begin{qprognamed}{cnot_conjugation_5}
        qubit a
        qubit b
        def	ua,0,'R_X(\alpha)'
        cnot a,b
        ua b
        cnot a,b
        \end{qprognamed}
        =
        \begin{qprognamed}{cnot_conjugation_6}
        qubit a
        qubit b
        def	uc,0,'R_X(\alpha)'
        nop a
        uc b
        \end{qprognamed}
    \end{equation}
        \item 
    \begin{equation}
        \begin{qprognamed}{cnot_conjugation_7}
        qubit a
        qubit b
        def	ua,0,'R_X^-(\alpha)'
        cnot a,b
        ua a
        cnot a,b
        \end{qprognamed}
        =
        \begin{qprognamed}{cnot_conjugation_8}
        qubit a
        qubit b
        def	uc,0,'Z'
        def ud,0,'R_X^-(\alpha)'
        uc a
        ud b
        \end{qprognamed}
    \end{equation}
        \item 
    \begin{equation}
        \begin{qprognamed}{cnot_conjugation_9}
        qubit a
        qubit b
        def	ua,0,'H(\alpha)'
        def ub,0,'H(\beta)^\dag'
        cnot b,a
        ua a
        ub b
        cnot a,b
        \end{qprognamed}
        =
        \begin{qprognamed}{cnot_conjugation_10}
        qubit a
        qubit b
        def	uc,0,'H(\alpha)'
        def ud,0,'H(\beta)^\dag'
        uc a
        ud b
        \end{qprognamed}
    \end{equation}
        \item 
    \begin{equation}
        \begin{qprognamed}{cnot_conjugation_11}
        qubit a
        qubit b
        def	ua,0,'H^-(\alpha)'
        def ub,0,'H(\beta)^\dag'
        cnot b,a
        ua a
        ub b
        cnot a,b
        \end{qprognamed}
        =
        \begin{qprognamed}{cnot_conjugation_12}
        qubit a
        qubit b
        def	uc,0,'H^-(\alpha)'
        def ud,0,'H(\beta+\pi)^\dag'
        uc a
        ud b
        \end{qprognamed}
    \end{equation}
        \item 
    \begin{equation}
        \begin{qprognamed}{cnot_conjugation_13}
        qubit a
        qubit b
        def	ua,0,'H(\alpha)'
        def ub,0,'H^-(\beta)^\dag'
        cnot b,a
        ua a
        ub b
        cnot a,b
        \end{qprognamed}
        =
        \begin{qprognamed}{cnot_conjugation_14}
        qubit a
        qubit b
        def	uc,0,'H(\alpha+\pi)'
        def ud,0,'H^-(\beta)^\dag'
        uc a
        ud b
        \end{qprognamed}
    \end{equation}
        \item 
    \begin{equation}
        \begin{qprognamed}{cnot_conjugation_15}
        qubit a
        qubit b
        def	ua,0,'H^-(\alpha)'
        def ub,0,'H^-(\beta)^\dag'
        cnot b,a
        ua a
        ub b
        cnot a,b
        \end{qprognamed}
        =
        \begin{qprognamed}{cnot_conjugation_16}
        qubit a
        qubit b
        def	uc,0,'H^-(\alpha+\pi)'
        def ud,0,'H^-(\beta+\pi)^\dag'
        uc a
        ud b
        \end{qprognamed}
    \end{equation}
    \end{enumerate}
    \end{theorem}

It is easy to check that $CNOT$ conjugation rules and $CZ$ conjugation rules are equivalent to each other, by converting $CNOT$ to $CZ$ and vice versa.

\section{Parallel QAOA Decomposition}
\label{appendix:par_qaoa}
QAOA is one of the fashionable algorithms in NISQ era. We will use the QAOA program for solving MaxCut problems as our optimization test cases.

However, we face the problem of lacking commutativity when optimizing $QAOA$ programs: our device can't execute $U(B,\beta_i)$ operation directly and it has to be decomposed into basic gates according to Equation \ref{equ:qaoa_impl}, and the block-commutativity optimization chances by commutativity between $U(B, \beta_i)$ matrices are missed.

There have been different ways to optimize QAOA circuits with $U(B, \beta_i)$ commutable with each other in mind. For example, \cite{shi2019optimized} detects all two-qubit diagonal structures in the circuit and aggregate them, so that commutativity detection can be performed on aggregated blocks. Another layout synthesis algorithm (scheduling considering device layout) QAOA-OLSQ\cite{tan2020optimal} schedules QAOA circuits twice, the first time on a large granularity (named TB-OLSQ) and the second time on a small granularity (named OLSQ). The large-granularity pass allows block commutativity to be considered and gates are placed in blocks. The small-granularity pass finishes the scheduling. 

However, these two approaches both require the optimization algorithm to perform coarse-grain block-level scheduling in addition to fine-grain gate-level scheduling. We may want to find another way to give commutativity hints to a gate-scheduling algorithm without modifying the algorithm itself.

Equation \ref{equ:qaoa_impl} inspires us with the fact that the shape of decomposed form of $U(B, \beta_i)$ is a bit like $CNOT$ gate: it has a ``controller'' qubit and a ``controlled'' qubit; multiple blocks with the same ``controller'' qubit can be commuted and interleaved freely at gate level, and can be finished in 2 ticks on average instead of 3, as in Figure \ref{fig:qaoa_nearly_simultaneously}.

\begin{figure}[H]
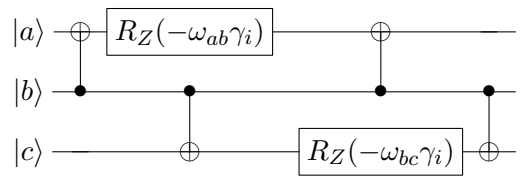

    \begin{qprognamed}{qaoa_nearly_simultaneously}
        qubit a
        qubit b
        qubit c
        def	ua,0,'R_Z(-\omega_{ab}\gamma_i)'
        def	uc,0,'R_Z(-\omega_{bc}\gamma_i)'
        cnot b,a
        cnot b,c
        ua a
        uc c
        cnot b,a
        cnot b,c
    \end{qprognamed}
    \caption{\label{fig:qaoa_nearly_simultaneously}The two blocks can be executed interleavingly.}
\end{figure}

The level of ``blocks'' according to the discovery above can be derived by directing and coloring all edges in the undirected graph $G=\langle V,E\rangle$:

\begin{itemize}
    \item First, we assign every edge with the direction in which we would perform the \ref{equ:qaoa_impl} decomposition (i.e. assign the graph with an orientation). Suppose the direction points from the controller qubit to the controlled qubit.
    \item Then, we colour all edges with minimal number of colours under the following constraints:
    \begin{enumerate}
        \item All in-degree edges of a vertex should be coloured differently from each other.
        \item Out-degree edges of a vertex should be coloured differently from all in-degree edges of the vertex.
    \end{enumerate}
\end{itemize}

The minimal number of required colors over all possible orientations is the minimal number of layers we can put these gates into.

Note that finding the minimal edge colouring under the constraints can be reduced to the problem of finding minimal vertex colouring of a new graph. In the new graph, vertices represent original edges; vertices for out-degree edges are fully connected; vertices for in-degree edges are connected with those for out-degree edges. 
Figure \ref{fig:qaoa_layer} is an example of assigning directions and colours for edges in the graph, and the equivalent vertex-colouring problem to the edge-colouring one.

\begin{figure}
    \begin{subfigure}{0.2\textwidth}
        \centering\includegraphics[width=\textwidth]{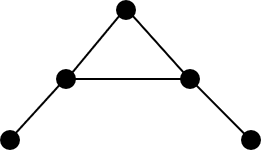}
        \caption{Graph for QAOA.}
    \end{subfigure}
    \hspace{1cm}
    \begin{subfigure}{0.2\textwidth}
        \centering\includegraphics[width=\textwidth]{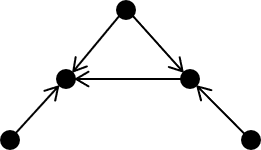}
        \caption{One orientation for the graph.}
    \end{subfigure}
    \begin{subfigure}{0.2\textwidth}
        \centering\includegraphics[width=\textwidth]{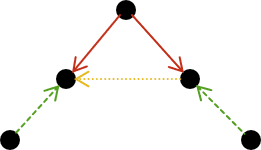}
        \caption{One coloring satisfying the constraints.}
    \end{subfigure}
    \hspace{1cm}
    \begin{subfigure}{0.2\textwidth}
        \centering\includegraphics[width=\textwidth]{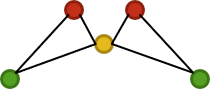}
        \caption{The equialent vertex-coloring problem.}
    \end{subfigure}
    \caption{\label{fig:qaoa_layer}Example for one possible orientation and layering of a graph.}
\end{figure}

One direct way to compute the block placement strategy is to use an SMT solver, for example, $QAOA-Par$ test cases in our evaluation are generated using Z3 Solver\cite{10.1007/978-3-540-78800-3_24}. We leave it as an open problem whether there is an efficient approach.

\section{Complexity Analysis}\label{complexity}

In this section we give a rough estimation of complexity of the
scheduling algorithm above. We put the main complexity results in table \ref{tab:complexity}, with some notes below to explain.

\begin{table}[]

\centering\begin{tabular}{|
>{\columncolor[HTML]{EFEFEF}}l |l|l|}
\hline
Step                & \cellcolor[HTML]{EFEFEF}Time & \cellcolor[HTML]{EFEFEF}Code Size \\ \hline
Compaction          & $O(n^2)$                     & $O(n)$                            \\ \hline
Unrolling           & $O(n^2+C^2n)$                & $C$ loops sized $O(Cn)$          \\ \hline
\multicolumn{3}{|c|}{\cellcolor[HTML]{EFEFEF}For each loop sized $O(m)$}               \\ \hline
Rotation            & $O(m^3)$                     & $O(m^3)$                      \\ \hline
Try $II$            & $O(logm)$                    & -                                 \\ \hline
Tarjan               & $O(m^2)$                     & -                                 \\ \hline
Floyd               & $O(m^3)$                     & -                                 \\ \hline
Scheduling          & $O(m^5)$                     & Span=$O(m^3)$           \\ \hline
Add $Z$             & $O(m^4)$                     & -                                 \\ \hline
Codegen             & $O(m^6)$                     & $O(m^3)$                          \\ \hline
Total               & $O(m^6)logm$                 & $O(m^3)$                          \\ \hline
\multicolumn{3}{|c|}{\cellcolor[HTML]{EFEFEF}In Total}                                 \\ \hline
Overall & $O(C^6n^6(logCn))$       & $O(C^4n^3)$                       \\ \hline
\end{tabular}

\caption{\label{tab:complexity}Complexity of our software pipelining approach.}
\end{table}

\subsection{Complexity of loop compaction}
Complexity for compacting a piece of loop program sized $O(n)$ once is $O(n^2)$, since when adding every instruction we check it against all instructions that are previously added.

\subsection{Complexity of loop unrolling}
Finding merging or cancelling candidates requires $O(n^2)$ time. Suppose the loop range is unknown, we have to perform the following steps on $C$ loops sized $m=O(Cn)$.

\subsection{Complexity of loop rotation}
A loop sized $O(n)$ can be rotated for at most $O(n^2)$ times, since loop rotation will not introduce new ``qubit'' into the loop, and the $O(n)$ qubits can be placed in an partial order: $q_a\prec q_b$ if a single qubit gate on $q_a$ will be on $q_b$ after rotation.

This will create a prologue sized $O(n^2)$, an epilogue sized $O(n^3)$ and a new loop sized $O(n)$. Each rotation requires $O(n^2)$ time (to find a rotatable gate) so the total complexity is $O(n^4)$.

\subsection{Complexity of modulo scheduling}
We need $O(logm)$ retries to binary-search the minimal $II$.
Complexity of Tarjan algorithm on a dense graph is $O(m^2)$, and complexity of Floyd algorithm is $O(m^3)$.

We leave the proof of complexity from retrying due to resource conflict in Appendix \ref{appendix:resource-scheduling-complexity}.

\subsection{Inversion pair detection}
The complexity for detecting in-loop inversion pair if $O(m^2)$.
The complexity for detecting across-loop inversion depends on the span of the total schedule. Note that according to Definition \ref{def:inversion}:
\begin{equation}
    r\leq (p_2-p_1) + \frac{c_1-c_1}{L},
\end{equation}
where $p_1,p_2=O(m^2)$. Thus
\begin{equation}
    r=O(m^2).
\end{equation}

The total complexity of checking $O(m^2)$ pairs of instructions across $r$ iterations is $O(m^4)$.

\subsection{Code generation}
The complexity for code generation is just the length of prologue and epilogue, $O(m^3)$. The compaction is of quadratic complexity so the total complexity is $O(m^6)$. However, for cases where the 
loop range is known, using a hash set to store the last \textbf{operation} on each qubit can reduce the complexity to $O(m^3)$.

\begin{theorem}
The total time complexity for our algorithm is
\begin{equation}
    O(C^6n^6(logCn)),
\end{equation}
and the size of the generated code is
\begin{equation}
    O(C^4n^3).
\end{equation}
\end{theorem}

\section{Adapting to existing architectures}
\label{sec:adapt_to_arch}
Note that we are building our approach of optimization based on a specific quantum circuit model as specified in Section \ref{subsec:arch}. Recall some of the features of the model that we use:
\begin{itemize}
    \item Classical computation and loop guards can be carried out instantly.
    \item The hardware can execute arbitrary single qubit operations and $CZ$ gates between arbitrary qubit pairs. All instructions can finish in one cycle.
    \item Instructions on totally different qubits can be carried out at the same time.
\end{itemize}

\subsection{Powerful classical control}

A quantum processor is usually split into classical part and quantum part, and all the classical logics (i.e. branch statements) are run on the classical part.

To implement fast classical guard for $for$-loops, we can use several classical architecture mechanisms, such as superscalar, classical branch prediction and speculative execution. As long as classical part commits instructions faster than quantum part executing instructions, we may keep the quantum part fully-loaded without introducing unnecessary bubbles.

If we want classical operations that affect the control flow of quantum part (e.g. classical branch statements), one way would be converting them to their quantum version. One practical example would be measurements with feedback: if we want to use the measurement outcome to control the following operations, we can just use a qubit array to replace classical memory, use $CNOT$ gate to replace measurement, and use controlled gate to replace classical control. The classical trick of register renaming can be adopted when converting measurement to quantum gates: different iterations can ``measure to'' different qubits to prevent unnecessary name dependency.

Also on real quantum processors the full-parallelism is not likely to be achieved, for example, there may be a limit of instruction issuing width on the device. For this case, we can just limit the maximal issuing width in resource conflict checking.

\subsection{CNOT-based instruction set}
One major difference between our assumptions and the real-world architectures is that most existing models and architectures adopt a $CNOT$-based instruction set, instead of a $CZ$-based one. We provide two possible approaches for extending our method to the $CNOT$-architecture case.

One approach is to convert the original circuit to $CZ$-version directly, using the equation $X[b]CZ[a,b]X[b]=CNOT[b]$. After optimization, an additional step is required to convert each $CZ$ gate into $CNOT$ gates by adding Hadamard gates. Note that the way of adding Hadamard gates can affect the depth of the kernel.
\begin{example}
Adding Hadamard gates on the same qubit of two adjacent $CZ$ gates saves gate depth by 1, compared to the version adding Hadamard gates on different qubits of the two $CZ$ gates.
\begin{equation}
    \begin{aligned}
    \begin{qprognamed}{cz2cnot_1}
    qubit a
    qubit b
    qubit c
    ZZ a,c
    ZZ b,c
    \end{qprognamed}
    =&
    \begin{qprognamed}{cz2cnot_2}
    qubit a
    qubit b
    qubit c
    H a
    cnot b,a
    cnot c,a
    H a
    \end{qprognamed}
    \\=&
    \begin{qprognamed}{cz2cnot_3}
    qubit a
    qubit b
    qubit c
    H a
    cnot b,a
    H a
    H c
    cnot a,c
    H c
    \end{qprognamed}
\end{aligned}
\end{equation}
\end{example}

{
However, deciding all directions of $CNOT$ gates can be a hard problem. We can formulate the problem as an ILP problem. A rough description is as follows:
\begin{itemize}
    \item Each $CZ$ is given a boolean variable, indicating the direction of $CNOT$ (and where to add Hadamard gates).
    \item If one $CZ$ is adjacent to a single qubit gate, the $H$ can be absorbed.
    \item If one $CZ$ is adjacent to another $CZ$ and if they add Hadamard on the same qubit, the two Hadamard can be cancelled and no depth is added.
    \item Otherwise the depth is added by $1$ from Hadamard. If there is an aliasing, the depth need to be added by more than $1$ so that $H$ gates on qubits with aliasing will be placed at two different ticks.
    \item The objective is to minimize the depth on all qubits.
\end{itemize}

We leave the best conversion from $CZ$ program into $CNOT$ program with minimal depth as a remaining problem.

Another way to port our approach is to modify our QDG definition to the $CNOT$-based instruction set. But in fact, the most commonly used $CNOT$ commutation rules that are based on intuition are only part of the complete $CNOT$ conjugation rules:

\begin{lemma}
    ($CNOT$ conjugation rules)\cite{ying2007algebra} There are 8 rules in total for $CNOT$ conjugation rules, similar to $CZ$ rules. See Appendix \ref{appendix:cnot_conjugation}.
\end{lemma}

If we want to exploit full power of these rules, we have to consider all these rules while building QDG, instead of considering only the intuitive rules (usually the first 4 rules).

But this time, the rewriting trick in Theorem \ref{th:generalized_conjugation} no longer works for $CNOT$ rules. How to use these rules directly for QDG construction remains an open problem.

}

\subsection{Working with device topology}
One problem about a controlled-Z architecture is that it can be hard to perform long-distance $CZ$ operation. For the $CNOT$ case, a long distance $CNOT$ gate with length $k$ can be implemented using $(4k-4)$ according to \cite{Shende_2006}. However, this is not true for $CZ$ gates, as ``amplitude'' can't propagate through $CZ$ gates.

A direct conversion approach can be taken by converting $CZ$ to $CNOT$ and back forth. Since every $CNOT$ is on critical path and no adjacent controlled bits can be found on critical path, this would require $(8k-8+1)=(8k-7)$ 
gates on critical path. The exception is $k=2$, since the last $Hadamard$ on the critical path should be removed and total depth is 8.

\section{Optimization of Cluster State Preparation, etc.}\label{cluster-optimization}

This chapter introduces the \textbf{Cluster} and \textbf{Array} test cases used in our evaluation.

\textbf{Cluster} is an example of cluster-state preparation program, which is a for-all loop: increasing count of iterations does not add to the overall depth of the program, which on the 2-dimensional grid is a constant $5$ ($4$ for $CZ$s in four directions and $1$ for Hadamard). Despite that, we can still perform loop optimization on this program to get a loop with kernel sized $1$.

\begin{figure}
\begin{subfigure}{0.08\textwidth}
    \centering\includegraphics[max width=\textwidth]{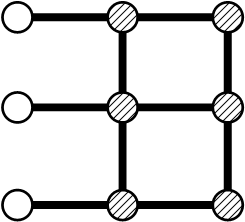}
    \caption{Before.}
    \label{fig:cs-outcome-1}
\end{subfigure}
\hspace{0.5cm}
\begin{subfigure}{0.35\textwidth}
    \centering
    \includegraphics[max width=\textwidth]{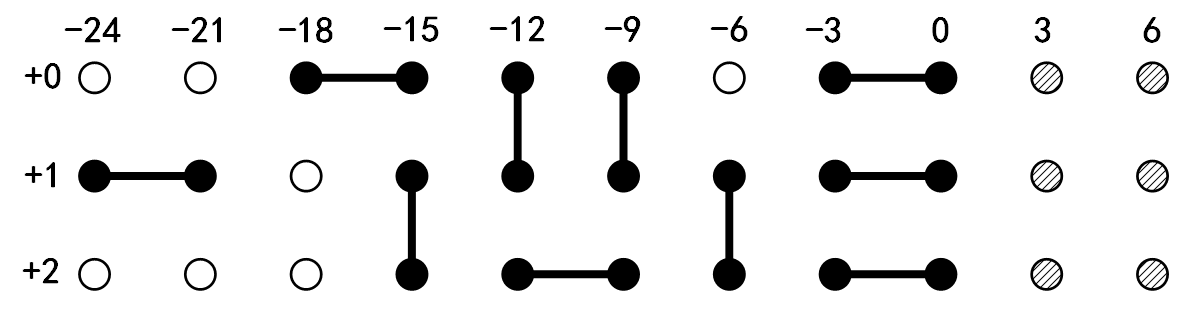}
    \caption{After. The numbers correspond to the intercept $b$ in expression $q[6i+b]$.}
    \label{fig:cs-outcome-2}
\end{subfigure}
\caption{Loop kernel for cluster state preparation ($N=3$). Shaded dots are qubits for Hadamard operands and closed dots are $CZ$ operands.}
    \label{fig:cs-outcome}
\end{figure}

For $C=2$, the loop kernels before and after rotation followed by software-pipelining is given in Figure \ref{fig:cs-outcome}. Our approach split $CZ$ gates that conflicts with each other into different iterations so that they can be executed together, and the kernel size is reduced to $1$, the best result for any loop-optimization approach except fully-unrolling.


\textbf{Array} series are several artificially-crafted loop programs on qubit arrays.
\textbf{Array 1} performs three $CZ$ gates as in Figure \ref{fig:resource_move}, while two Hadamard gates are added between $CZ$s to prevent cancellation.
\textbf{Array 2} performs non-cancelling $CZ$ gates so that they can be parallelized maximally.
\textbf{Array 3} constructs a huge Toffoli gate using Toffoli gates and ancillas: in each iteration, a Toffoli is performed on a source qubit, an ancilla and the next ancilla.

The instruction operands of these examples contain the iteration variable and are thus simpler to optimize compared with those on fixed set of qubits.

\end{document}